\documentclass[11pt]{article}
\usepackage{paralist, tabularx}
\usepackage{epsfig, graphicx}
\usepackage{multicol}
\usepackage{algorithmic,algorithm}
\usepackage{color}
\usepackage[]{hyperref}
\usepackage{smallhead}

\usepackage{times}
\usepackage{xspace}

\newcommand{\TD}{\mbox{{\sc TD}}}
\newcommand{\IG}{{\cal N}}

\newcommand{\RT}{n^{5/3}}

\newcommand{\priority}{{\sc Priority}\xspace}
\newcommand{\diff}{{\sc Diff}\xspace}
\newcommand{\randdiff}{{\sc Rand-Diff}\xspace}
\newcommand{\symdiff}{{\sc Sym-Diff}\xspace}
\newcommand{\SKB}{{\sc SKB}\xspace}
\newcommand{\cG}{{\cal G}\xspace}

\newboolean{short}
\setboolean{short}{false} % set this to true for short version

\newcommand{\shortOnly}[1]{\ifthenelse{\boolean{short}}{#1}{}}
\newcommand{\onlyShort}[1]{\ifthenelse{\boolean{short}}{#1}{}}
\newcommand{\longOnly}[1]{\ifthenelse{\boolean{short}}{}{#1}}
\newcommand{\onlyLong}[1]{\ifthenelse{\boolean{short}}{}{#1}}

\onlyLong{
\usepackage{amsmath,amssymb,amsthm}
\usepackage{fullpage}
}

\begin{document}
%Macros
\newcommand{\junk}[1]{}

\onlyLong{
\newtheorem{lemma}{Lemma}[section]
\newtheorem{theorem}{Theorem}[section]
}
\newtheorem{informaltheorem}[lemma]{Informal Theorem}
\newtheorem{informallemma}[lemma]{Informal Lemma}
\onlyLong{
\newtheorem{corollary}[lemma]{Corollary}
\newtheorem{definition}{Definition}[section]
\newtheorem{proposition}[lemma]{Proposition}
\newtheorem{question}{Question}
\newtheorem{problem}{Problem}
\newtheorem{remark}[lemma]{Remark}
\newtheorem{claim}{Claim}
\newtheorem{fact}{Fact}
}
\newtheorem{challenge}{Challenge}
\newtheorem{observation}{Observation}
\newtheorem{openproblem}{Open Problem}
\newtheorem{openquestion}{Open question}
\onlyLong{
\newtheorem{conjecture}{Conjecture}
}
\newtheorem{game}{Game}

\newcommand{\beq}{\begin{equation}}
\newcommand{\eeq}{\end{equation}}
\newcommand{\beas}{\begin{eqnarray*}}
\newcommand{\eeas}{\end{eqnarray*}}

\newcommand{\poly}{\mathrm{poly}}
\newcommand{\eps}{\epsilon}
\newcommand{\e}{\epsilon}
\newcommand{\polylog}{\mathrm{polylog}}
\newcommand{\rob}[1]{\left( #1 \right)} %Round Brackets
\newcommand{\sqb}[1]{\left[ #1 \right]} %square Brackets
\newcommand{\cub}[1]{\left\{ #1 \right\} } %curly brackets
\newcommand{\rb}[1]{\left( #1 \right)} %Round
\newcommand{\abs}[1]{\left| #1 \right|} %| |
\newcommand{\zo}{\{0, 1\}}
\newcommand{\zonzo}{\zo^n \to \zo}
\newcommand{\zokzo}{\zo^k \to \zo}
\newcommand{\zot}{\{0,1,2\}}

\newcommand{\en}[1]{\marginpar{\textbf{#1}}}
\newcommand{\efn}[1]{\footnote{\textbf{#1}}}

\newcommand{\prob}[1]{\Pr \left[ #1 \right]}
\newcommand{\dprob}[3]{\Pr \left[ #1 ; #2 \mbox{;} #3 \right]}
\newcommand{\dsumprob}[2]{P(#1,#2)}
\newcommand{\expect}[1]{\mathbb{E}\left[ #1 \right]}
\newcommand{\mindegree}{\delta}

\newcommand{\BfPara}[1]{\noindent {\bf #1}.}

\newcommand{\cost}{\mbox{cost}}
\newcommand{\GNS}[1]{\mbox{{\sc GNS}}(#1)}

\newcommand{\randsymdiff}{{\sc SymDiff}}

\newcounter{listCounter}
\newenvironment{Enumerate}{\begin{list}{\arabic{listCounter}.}{
        \usecounter{listCounter}\ListLengths}}{\end{list}}
\newenvironment{AlphaList}{\begin{list}{(\alph{listCounter})}{
        \usecounter{listCounter}\ListLengths\setlength{\labelwidth}{3ex}}}{\end{list}}

\newcommand{\ListLengths}{\setlength{\itemsep}{0ex}\setlength{\topsep}{1ex}\setlength{\partopsep}{0ex}}

%% No indentation for enumeration
\newenvironment{NoIndentEnumerate}{\begin{list}{\arabic{listCounter}.}{
        \usecounter{listCounter}\setlength{\leftmargin}{1em}\ListLengths}}{\end{list}}

%% No indentation for the bullet.  The text is indented.  Nominal spacing.
\newenvironment{myitemize}{\begin{list}{$\bullet$}{\setlength{\leftmargin}{1em}}}{\end{list}}

%% No indentation for the bullet.  The text is indented.  Little spacing.
\newenvironment{mylowitemize}{\begin{list}{$\bullet$}{\setlength{\leftmargin}{1em}\ListLengths}}{\end{list}}

%End macros

\onlyShort{
\title{Information Spreading in Dynamic Networks under Oblivious Adversaries\thanks{Supported in part by grants NSF CCF-1422715, NSF CCF-1535929, and ONR N00014-12-1-1001}}
\author{John Augustine\inst{1}
\and Chen Avin\inst{2}
\and Mehraneh Liaee\inst{3}
\and Gopal Pandurangan\inst{4}
\and Rajmohan Rajaraman\inst{3}
}

\institute{IIT Madras, Chennai 600036, India\\
\email{augustine@iitm.ac.in}
\and
Ben-Gurion University of the Negev, Beer-Sheva 84105, Israel\\
 \email{avin@cse.bgu.ac.il}
\and
Northeastern University, Boston, 02115, USA\\
\email{\{mehraneh,rraj\}@ccs.neu.edu}
\and
University of Houston, Houston, TX 77204, USA\\
\email{gopalpandurangan@gmail.com}
}
}

\onlyLong{
\title{Information Spreading in Dynamic Networks under Oblivious Adversaries\thanks{Supported in part by grants NSF CCF-1422715, NSF CCF-1535929, and ONR N00014-12-1-1001}}
\author{John Augustine\thanks{IIT Madras, Chennai 600036, India; Email: {\tt augustine@iitm.ac.in}}
\and Chen Avin\thanks{Ben-Gurion University of the Negev, Beer-Sheva 84105, Israel; Email: {\tt avin@cse.bgu.ac.il}}
\and Mehraneh Liaee\thanks{Northeastern University, Boston, 02115, USA; Email: {\tt mehraneh@ccs.neu.edu}}
\and Gopal Pandurangan\thanks{University of Houston, Houston, TX 77204, USA; Email: {\tt gopalpandurangan@gmail.com}}
\and Rajmohan Rajaraman\thanks{Northeastern University, Boston, 02115, USA; Email: {\tt rraj@ccs.neu.edu}}
}
}

\date{}
\maketitle

% !TEX root = dynamic.tex
\begin{abstract}

We study the problem of all-to-all information exchange, also known as
{\em gossip}, in dynamic networks controlled by an adversary that can
modify the network arbitrarily from one round to another, provided
that the network is always connected.  In the gossip problem, there
are $n$ tokens arbitrarily distributed among the $n$ network nodes,
and the goal is to disseminate all the $n$ tokens to every node.  Our
focus is on {\em token-forwarding} algorithms, which do not manipulate
tokens in any way other than storing, copying, and forwarding them.
Gossip can be completed in linear time in any static network, but an
important and basic open question for dynamic networks is the
existence of a distributed protocol that can do significantly better
than an easily achievable bound of $O(n^2)$ rounds.

In previous work, it has been shown that under adaptive adversaries
---those that have full knowledge and control of the topology in every
round and also have knowledge of the distributed protocol including
its random choices---every token forwarding algorithm requires
$\Omega(n^2/\log n)$ rounds to complete.  In this paper, we study
oblivious adversaries, which differ from adaptive adversaries in one
crucial aspect--- they are {\em oblivious to the random choices} made
by the protocol.  We consider \randdiff, a natural distributed
algorithm in which neighbors exchange a token chosen uniformly at
random from the difference of their token sets.  Previous work has
shown that starting from a distribution in which each node has a
random constant fraction of the tokens, \randdiff completes in
$\tilde{O}(n)$ rounds.  In contrast, we show that a polynomial
slowdown is inevitable under more general distributions: we present an
$\tilde{\Omega}(n^{3/2})$ lower bound for \randdiff under an oblivious
adversary.  We also present an $\tilde{\Omega}(n^{4/3})$ lower bound
under a stronger notion of oblivious adversary for a class of
randomized distributed algorithms---{\em symmetric knowledge-based
  algorithms}--- in which nodes make token transmission decisions
based entirely on the sets of tokens they possess over time.  On the
positive side, we present a centralized algorithm that completes
gossip in $\tilde{O}(n^{3/2})$ rounds with high probability, under any
oblivious adversary.
%We then show that every distributed neighbor-oblivious algorithm
%requires $\Omega(n^{3/2})$ rounds with high probability.  
We also show an $\tilde{O}(n^{5/3})$ upper bound for \randdiff in a
restricted class of oblivious adversaries, which we call {\em
  paths-respecting}, that may be of independent interest.
\junk{
Our results are a step towards understanding the true complexity of
information spreading in dynamic networks under oblivious adversaries.}

\junk{
\noindent {\bf Keywords:} Dynamic networks, Distributed Computation, Information
Spreading, Gossip, Randomization}
\end{abstract}

% !TEX root = dynamic.tex
%%% CHEN: Need to fix merge citations 

\section{Introduction}
\label{sec:intro}
% motivation dynamic networks
In a dynamic network, nodes (processors/end hosts) and communication
links can appear and disappear over time.  The networks of the current
era are inherently dynamic.  Modern communication networks (e.g.,
Internet, peer-to-peer, ad-hoc networks and sensor networks) and
information networks (e.g., the Web, peer-to-peer networks and on-line
social networks), and emerging technologies such as drone swarms are
dynamic networked systems that are larger and more complex than ever
before.  Indeed, many such networks are subject to continuous
structural changes over time due to sleep modes, channel fluctuations,
mobility, device failures, nodes joining or leaving the system, and
many other
factors~\cite{kempe02connectivity,sirocco,p2p-soda,cooper-frieze,flaxman04efficient,bollobas04the-diameter,liben-nowell05geographic}.
Therefore the formal study of algorithms for dynamic networks have
gained much popularity in recent years and many of the classical
problems and algorithms for static networks were extended to dynamic
networks.  During the past decade, new dynamic network models have
been introduced to capture specific
applications~\cite{broder00graph,bollobas04the-diameter,kempe02connectivity,ferreira04building,flaxman04efficient,odell+w:dynamic},
and the last few years have witnessed a burst of research activity on
broadcasting, flooding, random-walk based, and gossip-style protocols
in dynamic
networks~\cite{Kuhn2010Distributed,haeupler2010analyzing,Gurevich2009Correctness,Clementi2009Broadcasting,sarwate2009impact,baumann2009parsimonious,avin08explore,podc13,focs15,Clementi2008Flooding,Georgiou2008On-the-complexity,Augustine2013Storage,Dutta2013On-the-Complexity,casteigts2012time,Sarma2012Fast,baumann2011parsimonious,haeupler2011faster,kuhn2011coordinated}.

% motivation k-gossip
Our paper continues this effort and studies a fundamental problem of
information spreading, called {\em $k$-gossip}, on dynamic networks.
In $k$-gossip (also referred to as {\em $k$-token dissemination}), $k$
distinct pieces of information (tokens) are initially present in some
nodes, and the problem is to disseminate all the tokens to all the
nodes, under the constraint that one token can be sent on an edge per
round of synchronous communication.  This problem is a fundamental
primitive for distributed computing; indeed, solving $n$-gossip, where
each node starts with exactly one token, allows any function of the
initial states of the nodes to be computed, assuming the nodes know
$n$~\cite{Kuhn2010Distributed}.  This problem was analyzed for static
networks by Topkis~\cite{topkis:disseminate}, and was first studied on
dynamic networks for general $k$ in~\cite{Kuhn2010Distributed}, and
previously for the special case of one token and a random walk in
\cite{avin08explore}.

In this paper, we consider {\em token-forwarding} algorithms, which do
not manipulate tokens in any way other than storing, copying, and
forwarding them.  Token-forwarding algorithms are simple and easy to
implement, typically incur low overhead, and have been widely studied
(e.g, see~\cite{leightonbook,pelegbook}).  In any $n$-node {\em
  static} network, a simple token-forwarding algorithm that pipelines
tokens up a rooted spanning tree, and then broadcasts
them down the tree completes $k$-gossip in $O(n + k)$
rounds~\cite{topkis:disseminate,pelegbook}; this is tight since
$\Omega(n+k)$ is a trivial lower bound due to bandwidth
constraints.  A central question motivating our study is whether a
linear or near-linear bound is achievable for $k$-gossip on dynamic
networks.  It is important to note that algorithms that {\em
  manipulate}\/ tokens, e.g., network coding based algorithms, have
been shown to be efficient in dynamic
settings~\cite{haeupler2010analyzing}, but are harder to implement and
incur a large overhead in message sizes.

% motivation oblivious adversary and other type of adversaries 
%A crucial step in studying dynamic networks is the model for networks dynamic.
Several models have been proposed for dynamic networks in the
literature ranging from stochastic models
\cite{avin08explore,Clementi2008Flooding} to weak and strong adaptive
adversaries \cite{Kuhn2010Distributed}.  In this paper we consider one
of the most basic models known as the \emph{oblivious adversary}
\cite{avin08explore} or the \emph{evolving graph}
model~\cite{jarry04connectivity,ferreira07on-the-evaluation,ferreira04building,Sarma2012Fast}.
In this model, the adversary is unaware of any random decisions of the
algorithm/protocol and must fix the sequence of graphs before the
algorithm starts. The oblivious adversary can choose an arbitrary set
of communication links among the (fixed set) of nodes for each round,
with the only constraint being that the resulting communication graph
is connected in each round.  Formally, {\em oblivious adversary}\/
fixes an infinite sequence of connected graphs $\cG=G_1,G_2,\dots$ on
the same vertex set $V$; in round $t$, the algorithm operates on graph
$G_t$.  The adversary knows the algorithm, but is unaware of the
outcome of its random coin tosses.

%(This is equivalent to fixing the graphs in advance of the start of the algorithm.)
The oblivious adversary model captures worst-case dynamic changes that
may occur independent of the algorithm's (random) actions. On the
other hand, an {\em adaptive}\/ adversary can choose the communication
links in every round --- depending on the actions of the algorithm ---
and is much stronger. Indeed, strong lower bounds are known for these
adversaries\cite{Dutta2013On-the-Complexity,hkuhn}: in particular,
for the {\em strongly adaptive adversary}\footnote{In each round of
  the strongly adaptive adversary model, each node first chooses a
  token to {\em broadcast} to all its neighbors, and then the
  adversary chooses a connected network for that round with the
  knowledge of the tokens chosen by each node.}, there exists a
$\tilde{\Theta}(nk)$ lower bound\footnote{The notation
  $\tilde{\Omega}$ hides polylogarithmic factors in the denominator
  and $\tilde{O}$ hides polylogarithmic factors in the numerator.} for
$k$-gossip, essentially matching the trivial upper bound of $O(nk)$.

%An important open question is to understand the efficiency of information spreading and, in particular, $k$-gossip
%under an oblivious adversary.
%The current paper make several steps toward such understanding, but does not close all gaps.
% motivation randomize dif algorithm 
%The best lower bound currently known for dynamic networks and $n$-gossip are for
%strong adaptive adversaries and are of order $\tilde{\Theta}(n^2)$ \cite{Kuhn2010Distributed,Dutta2013On-the-Complexity}.
%% Chen: Should we discuss the results in more details?

The main focus of this paper is on closing the gap for the complexity
of $k$-gossip under an oblivious adversary between the straightforward
upper bound of $O(nk)$ and the trivial lower bound of
$\Omega(n+k)$\junk{\footnote{Understanding the complexity of oblivious
  adversary has been challenging in other settings and problems as
  well, e.g., in consensus \cite{barjoseph} and in Byzantine agreement
  in dynamic networks \cite{podc13}.}}.  In particular, can we achieve
an upper bound of the form $\tilde{\Theta}(n+k)$?  In fact, it is not
even clear whether there even exists a {\em centralized}\/ algorithm
that can do significantly better than the naive bound of
$O(nk)$.\junk{ (i.e., $O(n^{2-\epsilon})$ rounds (for any $\epsilon >
  0$)) under an oblivious adversary. In this work, we make progress
  towards answering this key question. We make progress towards
  answering this key question.}

The starting point of our study is \randdiff, a simple local
randomized algorithm for $k$-gossip.
%Each using a different feature from the algorithmic toolbox.
In each round of \randdiff, along every existing edge $(u,v)$ at that
round, $u$ sends a token selected uniformly at random from the
difference between the set of tokens held by $u$ and that held by node
$v$, if such a token exists.  Note that in \randdiff, a node is aware
of the tokens that its neighbours have and
%, and is \emph{not} what we call a \emph{neighbours oblivious} algorithm.
therefore \randdiff guarantees progress, i.e., exchange of a missing
token along \emph{every} edge where such a progress is possible.
\junk{It is also powerful in the sense that it may send
  \emph{different} tokens along different edges as opposed to a
  \emph{broadcast} model.}  Moverover, by using randomization it tries
to keep the entropy of token distribution as high as possible in the
presence of an adversary.  \randdiff is optimal for static networks,
while for dynamic networks under an oblivious adversary, it completes
$k$-gossip in $\tilde{O}(n +k)$ rounds for certain initial token
distributions which take any token-forwarding algorithm
$\tilde\Omega(nk)$ rounds under adaptive
adversaries~\cite{Dutta2013On-the-Complexity}\footnote{Actually,
  \cite{Dutta2013On-the-Complexity} shows the $O(n \ \polylog(n))$
  bound applies even for a weaker protocol called \symdiff, where the token
  exchanged between two neighbouring nodes is a random token from the
  {\em symmetric difference} of the token sets of the two
  nodes.\junk{\randdiff is more powerful than Sym-Diff and hence it can be
  shown that the same upper bound holds.}}.  \junk{More precisely, the
  work of \cite{Dutta2013On-the-Complexity} showed that starting from
  a distribution in which each node has a random constant fraction of
  the tokens, RandDiff completes in $O(n \ \polylog(n))$ rounds}
%On the other hand when randomization is not used, it was shown to preform poorly on dynamic networks.

\subsection{Our Contributions}
We present lower and upper bounds for information spreading under the
oblivious adversary model.

\smallskip
\noindent {\bf Lower Bound for \randdiff.}  We show that \randdiff
requires $\tilde{\Omega}(n^{\frac{3}{2}})$ rounds to complete
$n$-gossip under an oblivious adversary with high
probability\footnote{Throughout, by ``with high probability'' or {\em
    whp}, we mean with probability at least $1 - 1/n^c$, where the
  constant $c$ can be made sufficiently large by adjusting other
  parameters in the analysis.} (Section \ref{sec:rd}).  Our proof
shows that even an oblivious adversary can block \randdiff\ using a
sophisticated strategy that prevents some tokens from reaching certain
areas of the network. Although the adversary is unaware of the
algorithm's random choices, the adversary can exploit the
randomization of the algorithm to act against its own detriment.

\smallskip
\noindent {\bf Lower bound for symmetric knowledge-based algorithms.}
We use the technical machinery developed for the \randdiff\ lower
bound to attack a broad class of randomized $k$-gossip algorithms
called {\em symmetric knowledge-based (SKB) algorithms}, which are a
subclass of the knowledge-based class introduced
in~\cite{Kuhn2010Distributed} (Section \ref{sec:skb}).  In any round,
the token sent by a node in a knowledge-based algorithm is based
entirely on the set of tokens it possesses over time; an SKB algorithm
has the additional constraint that if two tokens first arrived at the
node at the same time, then their transmission probabilities are
identical.  SKB algorithms are quite general in the sense that each
node can use any probabilistic function that may depend on the node's
identity and the current round number to decide which token to send in
a round.  Indeed, this offers an attractive algorithmic feature that
does not exist in \randdiff: exploitation of information on the
history of token arrivals.  We show that this may not help achieve a
near-linear bound: any SKB algorithm for $n$-gossip requires
$\tilde{\Omega}(n^{\frac{4}{3}})$ rounds whp, under a stronger kind of
oblivious adversary, which is also allowed to add tokens from the
universe of $n$ tokens to any node in any round.

%\noindent {\bf Upper Bounds.}
We do not know whether either of the above lower bounds is tight.  Our
bounds do raise some intriguing questions: Can $n$-gossip be even
solved in $O(n^{2-\epsilon})$ rounds (for some constant $\epsilon >
0$) rounds by any algorithm?  Are there restricted versions of the
oblivious adversary that are more amenable to distributed algorithms?
We present two upper bound results that partially answer these
questions.

\smallskip
\noindent{{\bf Upper bound for \randdiff\ under restricted oblivious
    adversaries.}} We introduce a new model for dynamic networks which
restricts the oblivious adversary in the extent and location of
dynamics it can introduce (Section \ref{sec:ft}).  In the {\em
  paths-respecting} model, we assume that in each round, the dynamic
network is a subgraph of an an underlying {\em infrastructure graph}
$\IG$; furthermore, for every pair $(s,d)$ of nodes in $\IG$, there
exists a set $N_{sd}$ of simple vertex-disjoint paths from $s$ to $d$
in $\IG$ such that in any round the adversary can remove at most
$N_{sd} - 1$ edges from these paths.  The paths-respecting model is
quite general and of independent interest in modeling and analyzing
protocols for dynamic networks.\footnote{Indeed, an
  infrastructure-based model captures many real-world scenarios
  involving an underlying communication network with dynamics
  restricted to the network edges. This is unlike the case of a
  general oblivious adversary where the graph can change arbitrarily
  from round to round.}  A basic special case of the paths-respecting
model is one where $\IG$ is a $\lambda$-vertex-connected graph and the
adversary fails at most $\lambda-1$ edges in each round.  Even for
this special case, it is not obvious how to design fast distributed
algorithm for $n$-gossip.  In Section~\ref{sec:ft}, we also present
examples in this model where the adversary can remove a constant
fraction of the edges of an infrastructure graph.  We show that
\randdiff completes $n$-gossip in $\tilde{O}(n^{5/3})$ rounds under
the {\em paths-respecting model} (Section \ref{sec:ft}).  From a
technical standpoint, this result is the most difficult one in this
paper; it relies on a novel delay sequence argument, which may offer a
framework for other related routing and information dissemination
algorithms in dynamic networks.

\smallskip
\noindent{{\bf A $\min\{nk, \tilde{O}((n+k)\sqrt{n})\}$ centralized
    algorithm for $k$-gossip.}}  Finally, we present a centralized
algorithm (cf. Section \ref{sec:centralized}) that completes
$k$-gossip in $\min\{nk, \tilde{O}((n+k)\sqrt{n})\}$ rounds (and hence
$n$-gossip in $\tilde{O}(n^{\frac{3}{2}})$ rounds) whp, under an
oblivious adversary.  This answers the main open question
affirmatively, albeit in the {\em centralized} setting.  This result
provides the first {\em sub-quadratic} token dissemination schedule in
a dynamic network controlled by an oblivious adversary.  One of the
key ingredients of our algorithm is a load balancing routine that is
of independent interest: $n$ tokens are at a node, and the goal is to
distribute these tokens among the $n$ nodes, without making any copies
of the tokens. This load balancing routine is implemented in a
centralized manner; its complexity in the distributed setting under an
oblivious adversary, however, is open.  We believe that our
centralized algorithm is a step towards designing a possible
subquadratic-round fully distributed algorithm under an oblivious
adversary.

\onlyShort{ Due to space constraints, we have to omit many proofs; we
  refer the reader to the full paper for all missing proofs~\cite{augustine+alpr:dynamicFull}.}
\section{An $\tilde{\Omega}(n^{1.5})$ lower bound for \randdiff} 
\label{sec:rd}
%\maketitle

In this section, we show that there exists an oblivious adversary
under which \randdiff\ takes $\tilde{\Omega}(n^{3/2})$ rounds to complete
$n$-gossip whp.  
We will establish this result in
two stages.  In the first stage, we will introduce a more powerful
class of adversaries, which we refer to as {\em invasive}\/
adversaries.  Like an oblivious adversary, an invasive adversary can
arbitrarily change the graph connecting the nodes in each round,
subject to the constraint that the network is connected.  In addition,
an invasive adversary can add, to each node, an arbitrary set of
tokens from the existing universe of tokens.  Similar to an oblivious
adversary, an invasive adversary needs to specify, for each round,
the network connecting the nodes as well as the tokens to add to each
node, in advance of the execution of the gossip
algorithm.  

In Section~\ref{sec:lower.randdiff.invasive}, we will first show that
there exists an invasive adversary under which \randdiff\ takes
$\tilde{\Omega}(n^{3/2})$ rounds to complete $n$-gossip whp.
In Section~\ref{sec:lower.randdiff.oblivious}, we will simulate the
token addition process using \randdiff\ and extend the lower bound
claim to oblivious adversaries.

\begin{figure}[t]
\centering
\includegraphics[width=\textwidth]{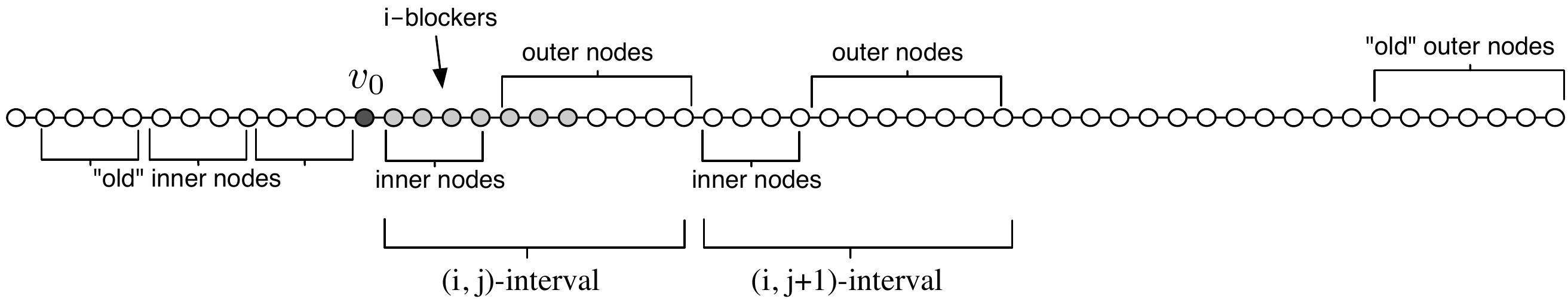}
\caption{The dynamic line network for the lower bound for \randdiff}
\label{fig:line-example}
\end{figure}

\onlyLong{
At a very high level, we use a dedicated set of $m$ tokens, for a suitable
choice of $m$, to block progress of an arbitrary token for a number of
rounds super-linear in $m$.  A judicious repetition of this process,
together with appropriate network dynamics, and a careful setting of
parameters then yields the desired lower bound.
}
\subsection{Lower bound under an invasive adversary}
\label{sec:lower.randdiff.invasive}

Our invasive adversary proceeds in $\sqrt{n}/(2 \log n)$ {\em phases},
each phase consisting of $\Omega(n)$ rounds, divided into {\em
  segments}\/ of $\sqrt{n}$ rounds each.  Throughout the process, the
network is always a line, you can refer to Fig. \ref{fig:line-example}
throughout the description on the network.  We build this line network
by attaching two line networks -- which we refer to as {\em left}\/
and {\em right}\/ lines -- each of which has the same designated
source node $v_0$ at one of its ends.  The size of the left line keeps
growing with time, while the size of the right line shrinks with time.
After the end of each segment, we move $\log n$ nodes closest to $v_0$
in the right line to the left line, so the size of the left line at
the start of segment $j$ of phase $i$ is exactly $((i-1)\frac{\sqrt{n}}{3} + j-1)\log n$.

At the start of each phase, we label the nodes in the right line
(other than the source $v_0$) as $v_1$ through $v_{p}$ (where $p$ is
the number of nodes in the right line at that time).  For any $j$, we
refer to set $\{v_l: 2(j-1)\sqrt{n} \le l < 2j \sqrt{n}\}$ as the
$(i,j)$-interval.  We refer to the first $\log{n}$ nodes of the
$(i,j)$-interval as the $(i,j)$-inner nodes, and the remaining
$2\sqrt{n}-\log{n}$ nodes as the $(i,j)$-outer nodes.

Initially, $v_0$ has all of the $n$ tokens and every other node has no
token.  We arbitrarily partition the \emph{tokens} into $\sqrt{n}$
groups of $\sqrt{n}$ tokens each.  We use $B_i$ to denote the $i$th
group, and refer to any token in $B_i$, $1 \le i \le \sqrt{n}/(2 \log
n)$ as an {\em $i$-blocker} since the adversary will use the tokens in
$B_i$ in phase $i$ to impede the progress of tokens not in $\cup_{j
  \le i} B_i$.  Let $M(u)$ denote the set of tokens in node $u$ at any
time.

At a high level, our adversary operates as follows.  Throughout
segment $j$ of phase $i$, the adversary keeps the line unchanged.  At
the start of segment $j$, the adversary adds randomly chosen subsets
of tokens from $B_i$ to the $\sqrt{n}$ nodes of $(i,j)$-interval which
are the $\sqrt{n}$ consecutive nodes adjacent to $v_0$ from the right.
We argue that this action ensures that in subsequent $\varepsilon
\sqrt{n}$ rounds, no token outside the set $\cup_{i'\le i}B_{i'}$
makes it to an $(i,j)$-outer node.  Since in each phase the adversary
uses the same set of $\sqrt{n}$ tokens, namely $B_i$, as ``blockers'',
it can continue this for $\Omega(\sqrt{n}/\log{n})$ phases, and ensure
that whp, no token in, say $B_{\sqrt{n}}$, has
reached the right line in $\Omega(n^{3/2}/\log n)$ rounds. We now formally describe how our adversary operates.

%\noindent {\bf Phase $i$, $1 \le i \le \sqrt{n}/(2 \log n)$:}
%\begin{itemize}
%\item
%{\bf Segment $j$, $1 \le j \le \sqrt{n}$:} The network is a line, that
%has two parts.  The first part is inner line with $v_0$ at one end,
%connected to all the $(i',j')$-inner nodes, where either $i' \le i$ or
%$i' = i$ and $j' < j$.  The second part is a line with $v_0$ at one
%end connected to $(i,j')$-intervals in sequence, $j' \ge j$, followed
%by $(i,j')$-outer nodes, $j' < j$.
%
%In the $r$th round of the segment:
%\begin{itemize}
%\item
%{\bf Segment Insertion:} Add tokens from the $r$th subgroup of
%$B_i$ to each node in the first half of the $(i,j)$-interval.
%
%\item
%{\bf Run:} Execute the given knowledge-based algorithm.
%\end{itemize}
%\end{itemize}

\smallskip
\noindent \textbf{Phase $i$, $1\leq i \leq \sqrt{n}/(2 \log n)$}:
\begin{mylowitemize}
\item
{\bf Segment $j$, $1 \le j \leq \sqrt{n}/3$}: The network is a line, that
has two parts.  The first part is the left line with $v_0$ at one end,
connected to all the $(i',j')$-inner nodes, where either \emph{A}: $i' < i$ or \emph{B}:
$i' = i$ and $j' < j$.  The second part is a line with $v_0$ at one
end connected to $(i,j')$-intervals in sequence, $j' \ge j$, followed
by $(i,j')$-outer nodes, $j' < j$.

\begin{compactitem}
\item 
{\bf Pre-Segment Insertion:} For each token $\tau$ in $B_i$ and each
node $v$ among the first $\sqrt{n}$ nodes of $(i,j)$-interval nodes:
adversary inserts $\tau$ in $v$ independently with probability $1/2$.

\item
{\bf Run:} Execute \randdiff\ for $\varepsilon \sqrt{n}$ rounds of
segment $j$.

\item
{\bf Post-Segment Shifting:} The adversary moves the $(i,j)$-inner nodes to the left line,
and the $(i,j)$-outer nodes to the right end of the line and connect the $(i,j+1)$-interval to $v_0$.
\end{compactitem}

\item
{\bf Post-Phase Insertion:} For every node in the right line, the
adversary inserts any token missing from $B_i$.
\end{mylowitemize}

%\begin{lemma}
%\label{lem:knowledge.segment}
%During segment $j$ of phase $i$, no token from $\cup_{i' > i}
%B_{i'}$ reaches an $(i,j)$-outer node whp.
%\end{lemma}

\junk{
\begin{lemma}
\label{lem:knowledge.phase}
At the start of phase $i$, every token in any node in the right line
is from $\cup_{i' < i} B_{i'}$
\end{lemma}
}

%\begin{theorem}
%Every knowledge-based algorithm requires $\Omega(n^{4/3})$ rounds for
%completing $n$-gossip under an invasive adversary.
%\end{theorem}
%\begin{proof}
%Each phase takes $\Omega(n)$ rounds, while the number of phases is $p
%= n^{1/3}/(2 \log n)$.  The number of outer nodes at the end of the
%last phase is at least $n/2$.  And the number of tokens in $\cup{i'
%  \le p} B_{i'}$ is at most $n/(2 \log n)$.  By
%Lemma~\ref{lem:knowledge.phase}, it follows that there is at least one
%token missing from every outer node at the end of the $p$ phases, whp,
%completing the proof of the theorem.
%\end{proof}

\begin{lemma}
\label{lem_diff_expected}
In every round of phase $i$ and segment $j$, for any of two adjacent nodes $u$ and
$v$ on the $(i,j)$-inner nodes,
the probability that $|M(u) - M(v)|$ is less than $\sqrt{n}/16$ is at
most $e^{-\Omega(\sqrt{n})}$.
\end{lemma}
\begin{proof}
Let $X$ be the random variable denoting the number of tokens node $u$
has but node $v$ does not have, at the start of segment $j$.  Clearly,
$X$ equals $\sum_{\tau \in B_i} I_{\tau}$, where $I_{\tau}$ is the
indicator variable for token $\tau$; $I_{\tau}$ is 1 if $u$ has token
$\tau$ and $v$ does not have $\tau$; otherwise it is 0.  Using
linearity of expectation, we obtain $E[X] = \sum_{\tau \in B_{i}}
E[I_{\tau}]$. Since the adversary adds each token to each node with
probability of 1/2 independently, we have $E[I_{\tau}] = 1/4$ and
$E[X] = \sqrt{n}/4$.  Using a standard Chernoff bound argument, we
obtain that the probability that $X \leq \sqrt{n}/8$ is
$e^{-\Omega(\sqrt{n})}$.  During the remainder of segment $j$, since
each node has two neighbors on the line, node $v$ may receive at most
$2\varepsilon \sqrt{n}$ new tokens. Thus, $|M(u) - M(v)|$ is at least
$\sqrt{n}/8 - 2 \varepsilon \sqrt{n}$ whp (for $\varepsilon
\le 1/{32}$, this difference is at least $\sqrt{n}/16$).
\end{proof}

\begin{lemma} 
\label{lem_not_beyond_logn}
In segment $j$ of phase $i$, the probability that any token in
$\cup_{i' > i} B_{i'}$ reaches an $(i,j)$-outer node is at most
$1/n^{9}$.
\end{lemma}
\begin{proof}
Let $\alpha$ be an arbitrary token in the set $\cup_{i' > i} B_{i'}$.
By Lemma~\ref{lem_diff_expected}, the probability that at an arbitrary
round token $\alpha$ is sent from one node to its adjacent node on
$(i,j)$-interval is at most $16/\sqrt{n}$. The probability
that token $\alpha$ goes further than $\log{n}$ steps during segment
$j$ (which is $\varepsilon \sqrt{n}$ rounds) is at most ${\varepsilon
  \sqrt{n} \choose \log{n}} {(\frac{16}{
    \sqrt{n}})}^{\log{n}}$, which is $O(1/n^{10})$.  Now using union
bound, we obtain that the probability that any token in $\cup_{i' > i}
B_{i'}$ reaches any i-outer node is at most $n/n^{10} = 1/n^9$.
 \end{proof}

\begin{lemma} 
\label{lem:rightline.phase}
At the end of phase $i$, the set of tokens in any node $\neq v_0$ in the right
line is $\cup_{i' \le i} B_{i'}$ whp.
\end{lemma}
\begin{proof}
The proof is by induction on $i$.  For convenience, we set the
induction base case to be $i = 0$ and assume $B_0$ is the empty set.
So the base case, at the start of the algorithm, is trivial since
initally every node other than $v_0$ has no tokens.  For the induction
step, we consider phase $i > 0$.  Let $R_i$ denote the set of nodes in
the right line at the end of phase $i$.  We first observe that $R_i
\subseteq R_{i-1}$.  By the induction hypothesis, it follows that the
token set at every node in $R_i$ at the end of phase $i-1$ is
precisely $\cup_{i' < i} B_{i'}$.  Furthermore, the adversary
guarantees that every node in $R_i$ has all tokens from $B_i$ at the
end of phase $i$.

It remains to prove that no token from $\cup_{i' > i} B_{i'}$ arrives
at any node in $R_i$ during phase $i$.  Our proof is by contradiction.
Let $v$ be the first node in $R_i$ to receive a token $\tau$ from
$\cup_{i' > i} B_{i'}$ in phase $i$.  Since $v$ is first such node, it
received $\tau$ from $v_0$ or from an $(i,j)$-inner node since $R_i$
is the union of the sets of all $(i,j)$-outer nodes.  Now, $v$ can be
connected to such an $(i,j)$-inner node only during segment $j$.  By
Lemma~\ref{lem_not_beyond_logn}, however, no $(i,j)$-outer node receives a
token from $\cup_{i' > i} B_{i'}$ whp.
\end{proof}

\begin{theorem}
Under the invasive adversary defined above, whp,
\randdiff\ requires $\Omega(n^{3/2}/\log n)$ rounds to complete
$n$-gossip.
\end{theorem}
\begin{proof}
Each phase consists of $\sqrt{n}/3$ segments, with each segment having
$\eps \sqrt{n}$ rounds.  So the total number of rounds after
$\sqrt{n}/(2 \log n)$ phases is $\Omega(n^{3/2}/\log n)$.  We obtain that after
$\sqrt{n}/(2 \log n)$ phases, the size of the left line is at most
$n/2$, implying that the right line has $\Omega(n)$ nodes.  By
Lemma~\ref{lem:rightline.phase}, whp, every node in
the right line is missing at least one token, completing the proof of
the theorem.
\end{proof}

\junk{Moreover, no token unless tokens of
first group goes beyond the distance of $\log{n}$ from node $s$. Note
that during each phase, the set of tokens the adversary adds to some
nodes of paths are the same.

At the end of this phase, the adversary adds the blocking tokens of
phase $i$ to every node in the network, so we can simply ignore them
because these set of tokens are no more sent and received by any node
on the network.}

\subsection{Lower bound under an oblivious adversary}
\label{sec:lower.randdiff.oblivious}
In this section, we extend the lower bound established in
Section~\ref{sec:lower.randdiff.invasive} to oblivious adversaries.
Thus, the adversary can no longer insert tokens into the network
nodes; the pre-segment insertion and post-phase insertion steps of the
adversary of Section~\ref{sec:lower.randdiff.invasive} are no longer
permitted.  We simulate these two steps using \randdiff\ and a
judicious use of (oblivious) network dynamics.  

\onlyLong{
We now describe how to
implement the token insertion process for \randdiff\ with an oblivious
adversary.

\noindent
{\noindent {\textbf{Pre-Segment Insertion:}}} The pre-segment
insertion step occurs at the beginning of every segment $j$ of every
phase $i$.  Our implementation varies depending on whether $j$ is $1$
or greater than $1$.  Let $X_{i,j}$ denote the set of $n^{1/2}$ nodes
in the $(i,j)$-interval that are nearest to $v_0$.

\begin{mylowitemize}
\item
{\bf Token insertion for first segment of phase $i$:} To implement
token insertion, we add two new rounds to the first segment.  In the
first round, the adversary adds an edge from $v_0$ to each node in
$X_{i,1}$ (we refer to these as {\em direct}\/ edges).  For the second
round of the phase, the adversary removes the direct edges added above
(except the one that connects $v_{0}$ to its right neighbor), adds
back the edge from $v_0$ to its neighbor in the right line, and adds
an edge between any two nodes in $X_{i,1}$, independently with
probability $1/2$.  The remainder of the network is unchanged from the
previous round.  The remainder of the first segment follows exactly
the process of the invasive adversary, as in
Section~\ref{sec:lower.randdiff.invasive}.  Consistent with our
notation for the invasive adversary, we define $B_i$ to be the set of
tokens that were transferred from $v_0$ to $X_{i,1}$ in the first
round of phase $i$ (the $i$-blockers).

\item
{\bf Token insertion for remaining segments of phase $i$:} Between two
consecutive segments $j$ and $j+ 1$, $j \ge 1$, of a phase $i$, we
need a mechanism to transfer the set $B_i$ of tokens introduced into
the nodes in the set $X_{i,j}$ to the nodes in $X_{i,j+1}$.  The
oblivious adversary achieves this in two steps.  First, in addition to
the line network, it forms a clique for $\Theta(\log n)$ rounds among all the outer nodes 
in $X_{i,j}$.  Then, for one round, the adversary adds a
biclique between all the outer nodes in $X_{i,j}$ and all nodes in $X_{i,j+1}$.
\end{mylowitemize}
\noindent
} %onlyLong
\onlyLong{
We now show that the token distributions inserted by the invasive
adversary of Section~\ref{sec:lower.randdiff.invasive} are achieved by
the actions of the above oblivious adversary.  We begin by showing
in Lemmas~\ref{lem:oblivious.pre-segment.first.first}
through~\ref{lem:oblivious.pre-segment.first.second} that the pre-segment
insertion step for the first segment of each phase is faithfully
implemented.
}%onlyLong
\junk{ %we repeat this paragraph
We now show that the token distributions inserted by the invasive
adversary of Section~\ref{sec:lower.randdiff.invasive} are achieved by
the actions of the above oblivious adversary.  We show
in Lemmas~\ref{lem:oblivious.pre-segment.first.first}
through~\ref{lem:oblivious.pre-segment.rest} that the pre-segment
insertion steps for all the segments are faithfully
implemented.
}

\junk{
So during this one round, each
of those $n^{1/3}$ nodes has exactly one token (Remember that we
ignore the $m_i$ tokens each node on the network has).
}

\onlyLong{
\begin{lemma} 
\label{lem:oblivious.pre-segment.first.first}
The number of tokens in $B_i$ is $n^{1/2}(1 - o(1))$ whp.
\end{lemma}
\begin{proof}
Let $M_i$ denote the complement of set $\cup_{i' < i} B_{i'}$.  For any token
$\tau$ in $M_i$, let $I_{\tau}$ be an indicator random variable that
is 1 if token $\tau$ is sent to a node in $X_{i,1}$ through any of the
$n^{1/2}$ links from $v_{0}$, and $0$ otherwise.  We thus have $E[I_\tau]
= 1 - (1-\frac{1}{|M_i|})^{n^{1/2}}$.  Then, the expected size of
$B_i$, that is, the expected number of different tokens that the nodes
in $X_i$ together collect in the first round of phase $i$ is
\[
|M_i| \left( 1 - \left(1-\frac{1}{|M_i|}\right)^{n^{1/2}}\right) = n^{1/2}(1 - o(1)),
\]
since $|M_i|$ exceeds $n/2$ for every phase.

Though the $I_\tau$'s, for different tokens $\tau$, are not
independent of one another, a standard application of the method of
bounded differences and Azuma's inequality yields that the size of
$B_i$ is $n^{1/2}(1 - o(1))$ whp.  \junk{ $X =
  \sum_{i=1}^{n-m_j}{I_{i}}$ and $E[X] = \sum_{i=1}^{n-m_j}
      {E[I_{i}]}$. The probability of $I_{i} = 1$ is for any i. Thus,
      we have the following equation for $E[X]$:
\begin{eqnarray*} \nonumber
 E[X] = 
 \\ \leq  (n-m_j)(\frac{n^{1/3}}{n-m_j}) = n^{1/3}
\end{eqnarray*}
}
\end{proof}
} % onlyLong

\junk{
So at the end of this one round, each of $n^{1/3}$ of nodes on first path has exactly one token, and the set of union of their tokens has $n^{1/3}$ tokens whp. Let's call these tokens blocking tokens of phase $i$. 
\item{Step 2: distributing blocking tokens}\\

For the next round, adversary first removes edges from those $n^{1/3}$ nodes on first path to $s$, instead it takes the first closest $n^{1/3} + \log{n}$ non-source nodes on first path, and adds a link between each pair of those nodes with probability of 1/2. Since the network should be connected, it adds a link from a source node to a non-source node on first path. 
 Then, it keeps the rest of network unchanged from previous round. We can show lemma \ref{lem_diff_expected_azuma} holds at the end of this round. 
}

\onlyLong{
\begin{lemma}
\label{lem:oblivious.pre-segment.first.second}
Suppose $|B_i|$ is $n^{1/2}(1 - o(1))$.  Whp over
the random choices in the first round, for any token $\tau$ in $B_i$,
and any node $v$ in $X_{i,1}$, $\tau$ is in $v$ at the end of round with
probability at least $1/2$ and at most some constant $p < 1$,
independent of every other token in $B_i$.
\end{lemma}
\begin{proof}
The probability that more than $c$ copies of a token exist in $X_{i,1}$
after the end of the first round is at most $\binom{n^{1/2}}{c}
(1/|M_i|)^c$, which can be made an arbitrarily small
inverse-polynomial by setting $c$ suitably high, since $|M_i| \ge
n/2$.  

Now, fix a node $v$ in $X_{i,1}$ and a token $\tau$ in $B_i$.  Let $\ell$
denote the number of copies of $\tau$ in the nodes of $X_{i,1}$ at the end of
the first round of phase $i$.  The probability that $\tau$ is in $v$
at the end of the second round is at least $1/2$ (since $\ell \ge 1$)
and at most $1 - 1/2^c$.
\end{proof}
}
% onlyLong

\onlyLong{
We now show that the pre-segment step for the remaining segments of
each phase are implemented faithfully by the oblivious adversary.

\begin{lemma}
\label{lem:oblivious.pre-segment.rest}
After the $1 + \log n$ rounds introduced by the oblivious adversary
between segments $j$ and $j+1$, the probability that a given token
$\tau$ in $B_i$ is at a given node $v$ in $X_{i,j+1}$ is a positive
constant in $(0,1)$, independent of every other token in $B_i$.
\end{lemma}
\begin{proof}
First, the clique over $\log n$ rounds guarantees that there is a
coupon collector process for each outer node in $X_{i,j}$, so that whp, every outer node in $X_{i,j}$ has every token in $B_i$ after
the $\Theta(\log n)$ rounds.  For the remainder of the proof, we assume that
the preceding condition holds.  

Fix $v$ in $X_{i,j+1}$ and $\tau$ in $B_i$.  In the next round, the
probability that $\tau$ is sent to $v$ is exactly $1 - (1 -
1/|B_i|)^{\Omega(\sqrt{n})}$.  This probability is $e/(e-1) \pm o(1)$ since
  $|B_i|$ is $\Omega(\sqrt{n}(1 - o(1)))$ whp and every outer node in $X_{i,j}$ has
  all of $B_i$ whp before this round.  This completes the proof of the
  desired claim.
\end{proof}
} % onlyLong

\onlyLong{
{\noindent {\textbf{Post-Phase Insertion:}}} At the end of phase $i$,
  our oblivious adversary simulates the insertion process of the
  invasive adversary, using one round: in addition to the network
  edges present in the last round of the last segment of the phase,
  the adversary adds a clique over all nodes in the right line of $G$,
  excluding $v_0$.

Finally, we show that the post-phase insertion completes correctly whp, and prove the main result.

\begin{lemma} 
Whp, every node in the right line has every token in $B_i$
at the end of phase $i$.
\end{lemma}
\begin{proof}
The proof is by induction on phases.  For convenience, we set the base
case to $i = 0$ with $B_0$ being the empty set; so the claim is
trivially true.  We now consider the induction step, which concerns
phase $i$.  Fix arbitrary token $\tau$ in $B_i$ and an arbitrary $u$
in the right line.  We argue that if $u$ does not have $\tau$ prior to
the last round of phase $i$, then the probability that node $u$ does
not receive token $\tau$ in the last round is at most
$e^{-\Omega(\sqrt{n})}$.

By Lemma~\ref{lem:oblivious.pre-segment.rest}, there are at least $\sqrt{n}-\log{n}$
nodes among the nodes in $(i,j)$-interval that receive $\tau$ with at
least a constant probability in the pre-segment token insertion
process.  So whp, $\Omega(n)$ nodes in the right line have token
$\tau$ before the last round of the phase is executed.  Consider the
clique among the nodes of the right line in the last round of the
phase.  By Lemma~\ref{lem_not_beyond_logn} and the induction
hypothesis, the size of the difference of the sets of tokens in two
neighboring nodes in the right line is at most $n^{1/2}$.  So the
probability that token $\tau$ is not sent by one of these $\Omega(n)$
nodes to a node missing $\tau$ is at most $1-1/\Theta(n^{1/2})$.

Since each link is independent, the probability that token $\tau$ is
not sent through any of these links is at least
$(1-1/n^{1/2})^{\Omega(n)} = e^{-\Omega(n^{1/2})}$.  Applying a union
bound, we obtain that the probability that a token from $B_i$ is
missing at any node in the right line at the end of phase $i$ is at
most $n^{3/2}e^{-\Omega(n^{1/2})}$, completing the proof of the
desired claim.
\end{proof}

Thus, we can claim the following theorem.
} % onlyLong
\begin{theorem}
\randdiff\ requires $\Omega(n^{3/2}/\log n)$ rounds whp under an
oblivious adversary.
\end{theorem}

% !TEX root = dynamic.tex

\subsection{Lower bound for symmetric knowledge-based algorithms}
\label{sec:skb}
In this section, we present a lower bound for a broad class of
randomized algorithms for gossip, called {\em symmetric
  knowledge-based (SKB)}\/ algorithms.  We first introduce some
notation.  For round $t$, we define $a_t: U \times V \rightarrow T$,
where $U$ is the universe of all tokens and $V$ is the set of all
nodes: if $\tau$ is at $u$ at the start of round $t$, then $a_t(\tau,
u)$ is the time that $\tau$ first arrived at $u$; otherwise $a_t(\tau,
u)$ is $\bot$.

\begin{definition} 
An SKB algorithm is specified by a collection of functions $P_{t,u}: U
\rightarrow [0,1]$, where $P_{t,u}(\tau)$ is the probability with
which $u$ sends $\tau$ to each of its neighbors in round $t$,
satisfying the following properties:
\begin{mylowitemize}
\item
{\bf Token transmission:} for any $t$, if $a_t(\tau, u) =
\bot$, then $P_{t,u} = 0$, the different token sending events for a
node in round $t$ are mutually exclusive, and $\sum_{\tau \in U}
P_{t,u}(\tau) \leq 1$.
\item
{\bf Symmetry:} for any $\tau_1, \tau_2$ such that $a_t(\tau_1,u) =
a_t(\tau_2,u)$, $P_{t,u}(\tau_1) = P_{t,u}(\tau_2)$.
\end{mylowitemize}
\end{definition}
We note that the $P_{t,u}$ may differ arbitrarily from node to node
and round to round.  The symmetry property and the resulting
dependence on the arrival times of tokens are the only constraint on
the algorithm.

We now show that there exists an invasive adversary
under which \SKB takes $\Omega(\frac{n^{4/3}}{\log{n}})$ rounds to
complete $n$-gossip whp.  In order to block the progress of an
arbitrary token, the adversary inserts a subset of $m$ tokens, for a
suitable choice of $m$, at the same time as that token reaches a
node. We refer to this subset of tokens as a \textbf{Blocker Set}.  A
random selection of the blocker sets, a judicious repetition of this
process, together with appropriate network dynamics, yields the
desired lower bound.

\onlyLong{
At the start of the process, the invasive adversary takes
$\frac{n}{2\log{n}}$ of tokens arbitrarily, and forms
$\frac{n^{2/3}}{2\log{n}}$ blocker sets $B_{i,k}$ for $1\leq i \leq
\frac{n^{1/3}}{2\log{n}}$ and $1\leq k \leq n^{1/3}$, each consisting
of $n^{1/3}$ tokens.  Then, the adversary proceeds in
$\frac{n^{1/3}}{\log{n}}$ phases, each phase consisting of $\Omega(n)$
rounds, divided into $n^{2/3}$ segments. Through phase $i$, the
adversary uses blocker sets $B_{i,k}$ for $1\leq k \leq n^{1/3}$.

Throughout the process, the network is always a line, consisting of
three parts -- which we refer to as {\it left}, {\it middle} and {\it
  right}.  At the very beginning, the left part and right part are
empty, and all nodes of the network are included in the middle
part. The left most node of middle part is always called $s$ and has
all tokens in $U$.  The size of left line -- nodes at the left side of
node $s$ keeps growing with time.  \junk{ At the end of each segment,
  we move the first $ \log{n}$ nodes at the right side of node $s$, in
  the middle line to the left line, so the size of the left line at
  the start of segment $j$ of phase $i$ is $((i-1) n^{1/3} +
  j-1)\log{n}$.  \\\\ } \\\\

\noindent \textbf{Phase $i, 1\leq i \leq \frac{n^{1/3}}{2\log n}$:} At this
time, the left part has $(i-1) n^{2/3} \log{n} $ nodes, then the
adversary takes all the nodes at the right side of $s$ as the nodes in
the middle part, and makes the right part empty.

\begin{mylowitemize}
\item \textbf{Segment $j, 1 \leq j \leq n^{2/3}  $:} Segment $j$ is $n^{1/3}$ rounds. Let  $v_1, v_2, ..., v_{n^{1/3}}$ be the first $n^{1/3}$ nodes of middle part next to $s$, and call  $v_1, ..., v_{\log{n}}$ $(i,j)$-inner nodes and call $v_{\log{n}+1}, ..., v_{n^{1/3}}$ $(i,j)$-outer nodes. 
 During segment $j$, at round $k$ ($1\leq k \leq n^{1/3}$), the adversary inserts blocker set $B_{i,k}, B_{i,k-1}, ..., B_{i,1} $ to nodes $v_1, v_2, ... v_k$ respectively. 
\item {\textbf{Post-Segment Shifting $j$}:} The adversary takes the $(i,j)$-inner and $(i,j)$-outer nodes, then moves the first $\log{n}$ nodes of them to the left part, and moves rest of them, specifically $n^{1/3}-\log{n}$ nodes to the right part. 
\end{mylowitemize}

\begin{lemma} \label{segment_prob}
In phase $i$, segment $j$, no token in the set $\cup_{i'>i, 1\leq k \leq n^{1/3}} B_{i',k}$ reaches $(i,j)$-outer nodes. 
\end{lemma}
\begin{proof} 
Consider an arbitrary token $\tau$ in $\cup_{i'>i, 1\leq k \leq
  n^{1/3}} B_{i',k}$. Since segment $j$ is $n^{1/3}$ rounds, token
$\tau$ can go at most $n^{1/3}$ far from $s$. From the definition of
segment $j$, it follows that whenever a token $\tau$ reaches a node
$u$ in segment $j$, there is exactly one blocker set which is inserted
by adversary at the same round to the same node $u$. This implies that $P_{t',u}(\tau)
\leq \frac{1}{n^{1/3}}$, $t' \geq t$, assuming that arrival time of
$\tau$ at $u$ is $t$. The reason is that the algorithm cannot
distinguish between a token that has been inserted by the adversary
and a token that comes from the source. So the probability that token
$\tau$ goes one edge further is at most $\frac{1}{n^{1/3}}$, and the
probability that it goes beyond the $(i,j)$-inner nodes is at most

\[ \sum_{q=\log{n}}^{n^{1/3}} {n^{1/3} \choose q} (\frac{1}{n^{1/3}})^{q} (1 - \frac{1}{n^{1/3}})^{(n^{1/3}-q)} \leq {n^{1/3} \choose \log{n} } (\frac{1}{n^{1/3}})^{\log{n}} = o(\frac{1}{n^{10}}) \]

\end{proof}

\begin{lemma} \label{phase_prob}
Let $\tau$ be an arbitrary token from set  $\cup_{i'>p, 1\leq k \leq n^{1/3}} B_{i',k}$. Then after $p = \frac{n^{1/3}}{2 \log{n}}$ phases, with high probability token $\tau$ has not reached any of $(i,j)$-outer nodes, for $ 1 \leq i \leq p$, $1 \leq j \leq n^{1/3}$. 
\end{lemma}
\begin{proof} Using lemma \ref{segment_prob}, the probability that token $\tau$ reaches any $(i,j)$-outer node is as follows
\[
\frac{n^{1/3}}{2 \log{n}} \times n^{2/3} \times  \frac{1}{n^{10}} = o\left(\frac{1}{n^9}\right). 
\]
\end{proof}
}

\begin{theorem} 
Under an invasive adversary, \SKB requires
$\Omega(\frac{n^{4/3}}{\log{n}})$ rounds whp.
\end{theorem}

% !TEX root = dynamic.tex

%\newcommand{\randdiff}{\mbox{{\sc RandDiff}}}
%\newcommand{\TD}{\mbox{{\sc TD}}}
%\newcommand{\IG}{{\cal N}}
%\newcommand{\poly}{\mbox{poly}}

%\newcommand{\RT}{n^{5/3}}

\section{Analysis of \randdiff\ under a paths-respecting adversary}
\label{sec:ft}
In this section, first we introduce a new model, the {\em
  paths-respecting}\/ adversary, under which we show that
\randdiff\ completes $n$-gossip in $\tilde{O}(n^{5/3})$ rounds whp.

\subsection{The paths-respecting model}
In the paths-respecting model we assume that there is an underlying
infrastructure network $\IG$ such that at the start of every round
$t$, the network $N_t$ laid out by the adversary is a subgraph of
$\IG$; we refer to any edge in $\IG - N_t$ as an {\em inactive} or
{\em failed edge} in round $t$.  Before presenting the model, we note
that the assumption of an infrastructure network is essentially
without loss of generality.  For instance, it captures $1$-interval
connectivity, a central dynamic network model of Kuhn et
al~\cite{Kuhn2010Distributed}: we can let $\IG$ be the complete graph
and require that $N_t$ be a connected subgraph of $\IG$ for each $t$.

\begin{definition}
\label{def:flexibly}
The {\bf paths-respecting} model places some constraints on $\IG$ and
the set of edges that the adversary can render inactive in any given
round.  In particular, we assume that for every pair $(s,d)$ of nodes
in $\IG$, there exists a set $N_{sd}$ of simple vertex-disjoint paths
from $s$ to $d$ such that the total number of inactive edges of paths
in $N_{sd}$ in any round is at most $|N_{sd}| - 1$.
\end{definition}

Before analyzing the paths-respecting model, we present two examples.
First, a natural special case of this model is one where $\IG$ is a
$\lambda$-vertex-connected graph and the adversary fails at most
$\lambda-1$ edges in each round.  If $\lambda = 2$, then a simple
example is that of a ring network in which an arbitrary edge fails in
each round.  In this example, the adversary is significantly
restricted in the number of total edges it can fail in a given round;
yet, it is not obvious how a distributed token-forwarding algorithm
can exploit this fact since for any pair of vertices, no specific path
between the two may be active for more than $n$ rounds over an
interval of $\lambda n$ rounds.  A radically different example of the
paths-respecting model in which the adversary can fail {\em a constant
  fraction}\/ of edges in each round is the following: $\IG$ consists
of a set of $r$ {\em center}\/ vertices and a set of $n-r$ {\em
  terminals}, with an edge between each center and each other vertex.
Any two vertices have at least $r-1$ vertex-disjoint paths between
them.  An adversary can remove edges between $\lfloor (r-2)/2\rfloor$
of the centers and all the terminals -- and hence, nearly half of the
edges of the network -- while satisfying the constraint that at most
$r-2$ edges are removed in any collection of $r-1$ vertex-disjoint
paths passing through the centers.

Our main result here is the analysis of \randdiff\ in the
paths-respecting model.

\begin{theorem}
\label{thm:randdiff}
Under any $n$-node paths-respecting dynamic network,
\randdiff\ completes $n$-gossip in $O(\RT \log^3 n)$ rounds whp.
\end{theorem}

Our proof of Theorem~\ref{thm:randdiff} proceeds in a series of
arguments, beginning with a restricted version of the paths-respecting
model, and successively relaxing the restriction until we have the
result for the paths-respecting model.  Fix a token $\tau$, and source
$s$ that has $\tau$ at the start of round 0.  Let $d$ be an arbitrary
node in the network.  In our analysis, we focus our attention on the
set $N_{sd}$ of vertex-disjoint paths between $s$ and $d$ such that
the total number of inactive edges of $N_{sd}$ in any round is at most
$|N_{sd}| -1$.  In Section~\ref{sec:randdiff.uniform.one.inactive}, we
analyze \randdiff\ under the assumption that the lengths of all paths
in $N_{sd}$ are within a factor of two of one another, and the
adversary fails at most one edge in any path.  In
Section~\ref{sec:randdiff.uniform}, we drop the restrictions that at
most one edge is inactive in any path and path lengths are
near-uniform, and complete the proof of Theorem~\ref{thm:randdiff}.

\junk{We first present a basic lemma that applies to \randdiff, even when
the random choices are made by an adversary.

\begin{lemma} 
\label{lem:flexibly.diff}
In the paths-respecting dynamic network, for any arbitrary pair of
nodes $s$ and $d$, all tokens of $s$ at round 0 arrive at node $d$
within $O(|N_{sd}|n)$ rounds of running \randdiff.
\end{lemma}
\begin{proof}
Since the adversary is restricted to fail at most $|N_{sd}| - 1$
edges from paths in $N_{sd}$, at each round, there exists at least one
path in $N_{sd}$ that has no failed edge. Thus, during $ 2
|N_{sd}| n$ rounds, there exists one path $p$ among paths in
$N_{sd}$, which has no failed edges for at least $2n$ (not necessarily
consecutive) rounds. In the following, we use a delay sequence
argument to show that having path $p$ active during those $2n$ rounds
suffices for tokens in $s$ to arrive at $d$ within $O(|N_{sd}|n)$
rounds.

Consider an $r$-th round that path $p$ is active before the algorithm
ends, for some arbitrary $r>0$. We use a pebble $\pi$ and initially
put it on node $d$ and move it toward $s$ along with path $p$ while we
go backward in time. Let $v_t$ denote the location of pebble at round
$t$ on the path, and $u_t$ be its neighboring node on the path and
closer to $s$. Initially, $v_r = d$. Also $M_t$ denotes the set of
missing tokens associated with the pebble at time $t$. At each round t
that $p$ is active (it has no failed edge), if a token $\tau$ is sent
from $u_t$ to $v_t$, means that node $v_t$ does not have token $\tau$
at time $t-1$. if no token is sent towards $v_t$ from $u_t$, we can
move pebble to node $u_t$ knowing that the set of missing tokens of
$u_t$ at the end of round $t-1$ is a super set of $M_t$.

By this argument, at each round that path $p$ is active either the set
of missing tokens associated to the pebble is growing or the pebble is
moved one step closer to node $s$. This process can not continue for
more than $2n$ rounds since the distance between node $s$ and where
the pebble starts the process is at most $n$ and size of set of
missing tokens is at most $n$, thus we obtain that $r\leq 2n$. This
yields the proof of lemma.
\end{proof}
}
\subsection{Near-uniform length paths and at most one inactive edge per path}
\label{sec:randdiff.uniform.one.inactive}
\begin{lemma}
\label{lem:uniform_length_one_edge_per_path}
Suppose there exists an integer $l > 0$ such that the length of each
path in $N_{sd}$ is in $[l, 2l)$.  Further suppose that in addition to
  the conditions of the paths-respecting model, for every path in
  $N_{sd}$, the adversary can fail at most one edge in the path in any
  round.  Then, the token $\tau$ is at $d$ in $O(\RT \log{n})$ rounds
  whp.
\end{lemma}

The proof of Lemma~\ref{lem:uniform_length_one_edge_per_path} is a
{\em delay sequence argument}\/ that proceeds backwards in time.
Delay sequence arguments have been extensively used in the analysis of
routing algorithms~\cite{leighton:v1}.  A major technical challenge
we face in our analysis, distinct from previous use of delay sequence
arguments, is network dynamics.  The number of possible dynamic
networks, even subject to the paths-respecting model, is huge and
our analysis cannot afford to account for them independent of the
actions of the algorithm.  

\onlyLong{
\smallskip
\noindent{{\bf Pebbles.}}
We consider a run of the algorithm for $T = c \RT \log{n}$ rounds, for
a sufficiently large constant $c$.  We show that the probability that
$d$ is missing any tokens at the end of round $T$ is $1/\poly(n)$.  In
our analysis, we use a notion of ''pebbles moving along the paths".
Each pebble is a {\em certificate}\/ for the event that a node is
missing some tokens; in particular, each pebble has an associated set,
which represents a subset of the tokens missing at the node where the
pebble is located at that time.  Note that since the analysis proceeds
backwards in time, for any node (and any pebble) this set of missing
tokens grows during the course of the analysis.  We also remark that
we use the notion of pebbles for analysis only; pebbles should not be
confused with tokens being sent around the network by \randdiff.

\smallskip
\noindent{{\bf Pebble updates.}}  The way a pebble moves along a path
is as follows.  Suppose a pebble $\pi$ is located at node $v$ with
associated missing set $M$ at the end of round $t$.  Consider a
neighbor $u$ of $v$.  We now consider cases depending on what happened
in round $t$ of \randdiff\ along edge $(u,v)$.  If $(u,v)$ was made
inactive by the adversary, then we do not gain any more information
about missing tokens at $v$ at the end of round $t-1$.  If $(u,v)$ was
active, however, and no token was sent along $(u,v)$, then we know
that the set of tokens missing at $u$ at the end of round $t-1$ is a
superset of $M$; we can depict this case by having the pebble $\pi$,
with its associated missing set $M$, move to $u$; we call this a {\em
  pebble move}.  On the other hand, if $(u,v)$ was active and a token
$\alpha$ was sent along $(u,v)$, then we can depict this case by
setting the missing set of pebble $\pi$ at the end of round $t-1$ to
be $M+\{\alpha\}$; we call this a {\em missing set update}.  Thus, the
set of missing tokens associated with a pebble is monotonically
nondecreasing.

\smallskip
\noindent{{\bf Phases and segments.}}  Our analysis groups consecutive
rounds into phases, starting from round $T$ and proceeding backwards
in time.  At the start of each phase, we have a pebble located on each
path in $N_{sd}$.  We refer to these pebbles as {\em leader pebbles};
these start from $d$ and proceed toward $s$ during the course of the
analysis.  Let the leader pebble on the $i$th path of $N_{sd}$ be
labeled $\pi_i$, and the location of $\pi_i$ at the start of round $t$
be $v_i(t)$, and let $u_i(t)$ be the adjacent node to $v_i(t)$, on
path $i$, on the side closer to $s$.  Let $L$ denote the set of leader
pebbles.  At time $t$, let $M_i(t)$ denote the missing set associated
with the pebble $\pi_i$ and $S_i(t)$ denote the set of tokens of node
$u_i(t)$.  Each phase is divided into three segments, which are
described below.

We define the {\em total distance}\/ of the pebble set at the start of any
round $t$, $\TD(t)$, to be the sum, over all $i$, of the distance
between $v_i(t)$ and $s$; note that since paths are vertex-disjoint,
by definition, the $\TD(t)$ at any instant is at most $n$; our pebble
movement process will ensure that the total distance measure is
monotonically non-increasing with decreasing $t$.

\begin{figure}[t]
\centering
\includegraphics[width=\textwidth]{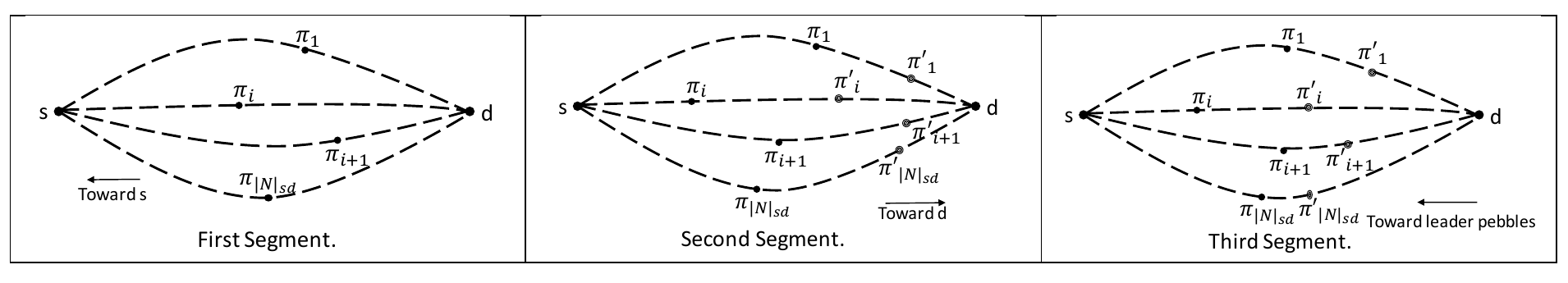}
\caption{Illustrating the three segments in the analysis of \randdiff}
\label{fig:segments}
\end{figure}

 %each segment consisting of$\Theta(n)$ rounds.
\smallskip
\noindent{{\bf Pebble Invariant.}}  We maintain the invariant that at
the start of any phase, the missing token sets that we associate with
the pebbles are all identical.  Consider the start of a phase in round
$t$.  Let $M$ denote the set of missing tokens at each of the pebbles.
At the beginning of the analysis, $t = T$, each of the $|N_{sd}|$
pebbles is located at node $d$, with the missing token set being the
set of all tokens missing at $d$ at the end of round $T$.  Thus, the
above invariant is satisfied at the start of the first phase.

\smallskip
\noindent{{\bf First Segment.}}  The first segment of any phase
consists of $T_1 = \max\{l,\frac{n}{l}\}$ rounds.  Recall that the
length of each path is in $[l, 2l)$; since the paths are
  vertex-disjoint, the number of paths is at most $n/l$.  If the
  number of times the leader pebbles move toward node $s$ during the
  $T_1$ rounds is at least $T_1/2$, then we let the second and third
  segments be empty, end this phase and proceed to the next phase.  In
  this case, to maintain the pebble invariant, for the beginning of
  next phase, we associate with each pebble the same set of missing
  tokens they had at the end of the last phase.  The second and third
  segments for the other case are described below, after the analysis
  of the first segment.  All segments are illustrated in
  Figure~\ref{fig:segments}.

\begin{lemma}[{\bf Analysis of first segment}]
\label{lem:segment_one}
Given any possible value for sets $S_i(t-T_1)$ and $M_i(t-T_1)$ of
tokens, and pebble locations at round $t-T_1$ such that $\TD (t-T_1)$
is greater than $\TD(t) - T_1/2$, we have $\bigcup_i M_i(t-T_1) \setminus M =
\Omega(T_1^{1/3}/\log{n})$ whp.  
\end{lemma}
\begin{proof}
In each of the $T_1$ rounds of the first segment, there exists at
least one path such that none of its edges are failed by the
adversary.  Therefore, in each round $t'\in [t - T_1, t]$ either at
least one leader pebble is moved toward node $s$ or at least one token
is sent from $u_i(t')$ to $v_i(t')$ for some path $i$, resulting in a
missing set update.  If the number of times the leader pebbles move
toward node $s$ during the $T_1$ rounds is at least $T_1/2$, then we
obtain that $\TD(t -T_1)$ is at most $\TD(t) - T_1/2$.  So in the
remainder of the proof, we assume that number of times that pebbles
move toward node $s$ is at most $T_1/2$.  Thus, in at least in half of
the rounds of the first segment, a token is sent from $u_i$ to $v_i$
for some $i$, resulting in a missing set update.  We consider two
separate cases:

\begin{itemize}
\item {\bf case a.} Suppose there exists a path $p_j \in N_{sd}$ such
  that at least $T_1^{1/3}$ of the $T_1/2$ token exchanges happen on path $j$.
  In this case, we infer the following for the size of set of missing
  tokens of $\pi_j$ during $[t-T_1, t]$:
\[ |M_j(t-T_1) - M| \geq T_1^{1/3}  |\bigcup_i M_i(t-T_1) - M| \geq T_1^{1/3}, \] 
completing the proof of the lemma in this case.

\item {\bf case b.} If there does not exist a path that has at least
  $T_1^{1/3}$ of the $T_1/2$ token exchanges, then there exists a set
  of paths $P$ of size at least $T_1^{2/3}/2$ such that at least one
  instance of token exchange happens on each of paths in $P$.  For
  each path in $P$ consider the instance which happens latest
  according to time during this $T_1$ rounds.  For path $p_i \in P$,
  let the latest instance of token exchange take place in time $t'_i
  \in [t-T, t]$. From the definition of \randdiff, node $u_i(t'_i)$
  sends a token $\alpha_{i}$ to $v_i(t'_i)$ from $S_i(t'_i) -
  M_i(t'_i)$ chosen uniformly at random; this implies that $S_i(t'_i)
  \cap M_i(t_i') \neq \emptyset$.  Note that $M_i(t') = M $ for $(t' >
  t'_i)$ since $t'_i$ is the latest time in segment 1, that an
  instance of token exchange has happened.  We also know that at $M
  \cup \{\alpha_i \} \subseteq M_i(t'_i)$ and $\alpha_i \notin M$.

We group the set $P$ of paths into at most $\log n$ categories: path
$p_i$ is in category $C_j$ if $S_i(t'_i) \cap M_i(t'_i)$ is in
range $[2^{j-1}, 2^{j})$.  One of these categories, say $C_j$, has
  $\Omega(T_1^{2/3}/\log(n))$ paths.  In the following, we show that
  $\bigcup_i \{\alpha_i\} = \Omega(T_1^{1/3}/\log^2(n))$ whp.  Let $k
  = 2^{j-1}$.  We say that a token has {\em high frequency}\/ if it
  appears in more than $ck\ln{n}$ different $S_i(t'_i) \cap
  M_i(t'_i)$, for a constant $c > 0$ to be specified later; otherwise,
  we call it as a token with {\em low frequency}.  If a token $\alpha$
  has high frequency, then $\alpha$ will be sent by some $u_j$ whp,
  since

\[ \Pr\{ \mbox{None of the nodes that has $\alpha$ send it}\} \leq \left(1-\frac{1}{2k}\right)^{ck \ln{n}} \le \frac{1}{n^{c/2}}  \] 

Every $S_i(t'_i) \cap M_i(t'_i)$ needs to have at least one token with
low frequency; otherwise, all its tokens will be sent by some $u_j$
whp.  Any token with low frequency will be picked $O(\log{n})$ times
whp, since the probability of sending each token in a given round is
at most $1/k$, and a token with low frequency by definition appears in
at most $ck\ln{n}$ of the sets.  Thus, by a Chernoff bound, the number
of distinct tokens sent by $S_i(t'_i) - M_i(t'_i)$ for all $i$ which
$p_i \in P$ is at least $\Omega(\frac{T_1^{2/3}}{k \log^2{n}})$ whp,
by just considering tokens with low frequencies.

If a particular $S_i(t'_i) \cap M_i(t'_i)$ has at least $k/2$ tokens
with high frequency, then the number of distinct tokens selected is at
least $k/2$; if there is no such case among any of $i$'s, then by the
above calculation using low frequency tokens, we obtain that the
number of distinct tokens selected is at least
$\Omega(T_1^{2/3}/\log^2{n})$ whp.  We thus obtain that the number of
distinct tokens sent is $\Omega(\max\{\frac{k}{2},
\frac{T_1^{2/3}}{k\log^2{n}}\})$, which is minimized or $k =
\sqrt{2}T_1^{1/3}/\log n $, yielding a bound of
$\Omega(\frac{T_1^{1/3}}{\log{n}})$.  Hence, we obtain that $\bigcup_i
M_i(t-T_1) - M = \Omega(T_1^{1/3}/\log{n})$ whp, completing the proof
of the lemma.
\end{itemize}
\end{proof}

\noindent{{\bf Second segment.}}  The second segment consists of $T_2$
rounds, where $T_2$ may vary from phase to phase.  At the start of the
second segment we introduce a new pebble $\pi'_i$ at each node $v_i$,
with associated missing token set $M_i$.  These pebbles proceed toward
$d$, updating their missing token sets as they move toward $d$.  Let
$t_m$ denote the time that it takes the last pebble $\pi'_m$ to arrive
at $d$.  Then, the missing token set at $d$ at this time is a superset
of the union, over $i$, of $M_i$, which has at least
$\Omega(T_1^{1/3}/\log(n))$ tokens more than $M$ whp, by
Lemma~\ref{lem:segment_one}.  

We consider two cases.  If $t_m$ is at most $2T_1^{1/3}$, then we set
$T_2 = t_m$, calling this an end to the second segment.  Otherwise, we
consider two subcases.  In the first subcase, the number of tokens in
the missing set associated with the pebble that arrived last is at
least $|M| + (t_m - T_1)/2$.  In this subcase, we set $T_2 = t_m$,
calling this an end to the second segment.  We obtain that the number
of missing tokens at $d$ is at least $|M| + (t_m - T_1)/2$.

In the second subcase, the number of tokens in the missing set
associated with the pebble that arrived last is less than $|M| + (t_m
- T_1)/2$.  This implies that the pebble $\pi'_m$ was blocked for at
least $(t_m - T_1)/2$ steps on its way to $d$, which in turn implies
that the pebble $\pi_m$ has increased its missing token set size by at
least $\frac{(t_m - T_1)}{2} - p$ tokens.  We now create a new copy of
pebble $\pi_m$ at $v_m$ (note that $v_m$ may have changed since the
start of the second segment) and send this pebble again toward $d$.
Again, we consider the time it takes for this new pebble to reach $d$.
If this is within $\Theta(T_1)$, or the arriving pebble gained tokens
at least a constant fraction of the time spent, we end the second
segment.  Otherwise, $\pi_m$ has gained tokens at least a constant
fraction of the time spent, in which case we repeat this argument.

We continue this until either the missing set of pebble $\pi_m$ is the
set of all tokens, or we find that $d$ is missing tokens whose size is
$|M| + \Omega(T_2)$.  It is easy to argue that under \randdiff, every
node receives at least one token in $O(n)$ rounds, so the first of the
two possibilities cannot happen.  We thus have the following lemma.

\begin{lemma}[{\bf Analysis of the second segment}]
\label{lem:segment_two}
If $T_2$ is the time taken for the second segment, then the set of
missing tokens at $d$ at time $t - T_1 - T_2$ has size at least $|M| +
\Omega(T_1^{1/3}/\log(n))$ if $T_2$ is $O(T_1)$, and at least $|M| +
\Omega(T_2)$ otherwise. \qed
\end{lemma}

\noindent{{\bf Third Segment.}}  In the third segment, we send pebbles
from $d$ to $v_i$ along each path $i$, so that the pebbles end up at
the same position as the start of the phase.  Again, we consider the
last time $t'_m$ at which the pebble on path say $m$ reaches $v_m$.
If $t'_m$ is $O(\max\{T_1,T_2\})$, then we terminate the third segment
and the phase.  In this case, by Lemma~\ref{lem:segment_two}, we
obtain at the end of this phase that for each $i$, the pebble at $v_i$
has missing set that has increased by size either
$\Omega(T_1^{1/3}/\log(n))$, if the time of the phase is $O(T_1)$, or
by at least a constant fraction times the length of the phase.

If $t'_m$ is $\Omega(\max\{T_1,T_2\})$, we find that during this
phase, the missing set of tokens at $v_m$ has increased by at least a
constant fraction times the length of the phase so far (but possibly
not at other nodes $v_i$).  As in the second segment, we send a pebble
from $v_m$ back to $d$; we repeat the argument, always having a node
whose number of missing tokens exceeds $|M|$ by a number that is at
least a constant fraction times the current duration of the phase.
This cannot go on for more than linear number of rounds, so it ends in
the situation where for all $i$, the pebble at $v_i$ has missing set
that has increased by size either $\Omega(T_1^{1/3}/\log(n))$, if the
time of the phase is $O(T_1)$, or by at least a constant fraction
times the length of the phase.

\begin{lemma}[{\bf Analysis of third segment}]
\label{lem:segment_three}
If $T_3$ is the time taken for the third segment, then each of the
pebbles at $v_i$ has an associated missing set $M'$ of tokens, where
$|M'|$ is at least $|M| + \Omega(T_1^{1/3}/\log(n))$ if $T_2$ + $T_3$
is $O(T_1)$, and at least $|M| + \Omega(T_2 + T_3)$ otherwise. \qed
\end{lemma}

We are now ready to complete the proof of
Lemma~\ref{lem:uniform_length_one_edge_per_path}.  Consider any phase
of length $L$.  If the phase ends after the first segment (hence, has
length $T_1$), we have a decrease in the total distance measure by half
the number of rounds in the phase.  Otherwise, by
Lemma~\ref{lem:segment_three}, we have whp that the number of missing
tokens associated with the pebbles at the end of the phase increases
by at least $T_1^{1/3}/\log n$, if $L = O(T_1)$, and $\Omega(T_1)$,
otherwise.  Hence, in a phase, either the rate of decrease of
total distance per round is at least $1/2$, or the rate of increase of the
number of missing tokens is at least $1/(n^{2/3} \log n)$.  Since the
total distance measure is initially $n$ and is always nonnegative, and the
number of missing tokens is initially 1 and is at most $n$, it follows
that $T$ is $O(n^{5/3} \log n)$ whp, completing the proof of
Lemma~\ref{lem:uniform_length_one_edge_per_path}.
}

\subsection{Removing restriction on path lengths and inactive edges per path}
\label{sec:randdiff.uniform}
We first extend the claim of the preceding section to the case where
the adversary can fail an arbitrary number of edges in any path of
$N_{sd}$, subject to the constraint imposed by the paths-respecting
model that the number of inactive edges in $N_{sd}$ is at most
$|N_{sd}| - 1$.  We continue to make the assumption of near-uniform
path lengths.  In a round, call a path {\em active}\/ if none of its
edges is failed, {\em 1-inactive} if exactly one of its edges is
inactive, and {\em dead} if more than one of its edges are inactive.
Since the adversary can fail at most $|N_{sd}| - 1$ edges among
$|N_{sd}|$ disjoint paths, it follows that the number of active paths
is at least one more than the number of dead paths.  This is the only
constraint we place on the adversary that we analyze in this section:
the number of active paths is at least one more than the number of
dead paths.

\begin{lemma}
\label{theo_const}
Suppose there exists an integer $l > 0$ such that the length of each
path in $N_{sd}$ is in $[l, 2l)$.  Further assume that the number of
  dead paths is in $[a,2a)$ for some $a$, in each round.  Then, under
    \randdiff, $\tau$ is at $d$ in $O(\RT \log{n})$ rounds whp.
    \junk{ Assume that in the setting with a strong adversary, the
      number of dead paths is any number between $a$ and $2a$ in each
      round.
%where $a$ is an integer such that $ 2a+1 \leq k$.
Then the process of sending $n$ tokens from $s$ to $d$ in a network
with paths of almost uniform length using \textbf{Rand-Diff} will
complete in $O(n\log{n})$ rounds whp.}
\end{lemma}
\onlyLong{
\begin{proof}
We consider the run of \randdiff\ for $T = c'\RT \log{n}$ rounds,
where $c'$ is a sufficiently large constant\iffalse $(c^{'} \geq c)$
\fi.  We use the probabilistic method~\cite{alon+s:prob} to show that
in the time interval $[1, T]$, there exists a set $R$ of $c\RT\log{n}$
(not necessarily consecutive) rounds in which there is a subset $P$ of
paths such that none of the paths from $P$ is dead and at least one of
paths in $P$ is active in each round in $R$.  (Here, the constant $c$
can be made sufficiently large by choosing $c'$ appropriately.)  We
then invoke Lemma~\ref{lem:uniform_length_one_edge_per_path} to
establish the desired claim.

All that remains is to establish the existence of $P$ and $R$ as
required above.  We choose a set $P$ of paths by picking each path
independently with probability $\frac{1}{a+1}$.  From the assumption
of the lemma, we know that any round has at most $2a$ number of dead
paths.  Therefore, the probability that none of these dead paths are
in $P$ is at least $(1-\frac{1}{a+1})^{2a} \geq \frac{1}{e^2}$.
By assumption, we have that in any round there exists at least $a+1$
active paths.  Therefore, the probability that there is at least one
active path among paths in $P$ in round $r$ is at least
$1 -(1 - \frac{1}{a+1} )^{a+1} \geq 1 - \frac{1}{e}$.

Since paths in $P$ are picked independently, the probability that in a
certain round, no path in $P$ is dead and at least one path in $P$ is
active is at least $\frac{1}{e^2}(1 - \frac{1}{e})$.  Thus, the
expected number of rounds that have no dead paths and have at least
one active path in $P$ is $c' \frac{e-1}{e^3} \RT\log{n}$. For $c'
\frac{e-1}{e^3} \geq c$, there exists a set $P$ of paths and set $R$
of at least $c\RT\log{n}$ rounds such that:
\begin{itemize}
\item for all $p$ in $P$ and $r$ in $R$, $p$ is either active or
  1-inactive during round $r$.
\item for each $r$ in $R$, there exists at least one path $p$ in $P$
  such that $p$ is active in round $r$.
\end{itemize}
This establishes the existence of $P$ and $R$ as desired, and
completes the proof of the lemma.
\end{proof}
} % onlyLong

Now, we extend Lemma~\ref{theo_const} by removing the constraint on
the number of dead paths.

\begin{lemma}
\label{lem:uniform_length}
Suppose there exists an integer $l > 0$ such that the length of each
path in $N_{sd}$ is in $[l, 2l)$.  Then, in the paths-respecting
  model, using \randdiff, the token $\tau$ is at $d$ in $O(\RT\log^2 n)$
  rounds whp.
\end{lemma}
\onlyLong{
\begin{proof}
We divide the $c' \RT \log^2 n$ rounds into $\log n$ classes, where the
$i$th class consists of rounds in which the number of dead paths is in
the interval $[2^i, 2^{i+1})$.  By simple averaging, we obtain that
    there is at least one class with at least $c'\RT\log{n}$ rounds, and
    the number of dead paths in any round in this class is in $[2^i,
      2^{i+1})$ for some integer $i > 0$.  We now apply
      Lemma~\ref{theo_const} to establish the desired claim.
\end{proof}
} % onlyLong

We now complete the proof of Theorem~\ref{thm:randdiff} \onlyLong{,
  restated below, } by removing the assumption of near-uniform path
lengths in the paths of $N_{sd}$.  This is a standard argument in
which we incur another multiplicative factor of $\log n$ in our bound.

\onlyLong{
\begin{theorem}
In the paths-respecting model, \randdiff\ completes gossip in
$O(\RT \log^3(n))$ rounds whp.
\end{theorem}
\begin{proof} 
Let $\tau$ be any token located at a source node $s$ at the start of
round $0$, and let $d$ be any other node.  Consider the set $N_{sd}$
of paths from $s$ to $d$ with the property that the number of edges which fail in paths
of $N_{sd}$ in any round is at most $|N_{sd}| - 1$.

We divide these $|N_{sd}|$ paths into $\log n$ groups such that group
$i$ ($1 \leq i \leq \log{n}$) includes paths of length between $2^{i}$
and $2^{i+1}-1$.  Let $\lambda_i$ shows the number of paths in group $i$.
By simple averaging, we obtain that in any round there is a group $i$
of paths such the number of edges removed by the adversary from
that group is at most $\lambda_i -1$ in that round.  Then if we consider
run of algorithm for $T = c'\RT\log^{3}{n}$ rounds, there is a group $i$
of paths such that for at least $c'\RT\log^2{n}$ rounds the number of
edges removed by the adversary from that group in each of these
rounds is at most $\lambda_i -1$.  We now apply Lemma~\ref{lem:uniform_length}
to derive that $d$ receives token $\tau$ in at most $T$ rounds whp.

Since the above claim holds for each token $\tau$, a union bound
yields us that \randdiff\ completes gossip in $O(\RT\log^{3}{n})$ rounds
whp.
\end{proof}
} % only Long

% !TEX root = dynamic.tex
\section{Centralized $k$-gossip in $\min\{nk, \tilde{O}((n+k)\sqrt{n})\}$ rounds} % against
                                                                                                                      % oblivious
                                                                                                                      % adversaries}
\label{sec:centralized}
In this section we present a centralized algorithm that completes
$k$-gossip in $\tilde{O}((n+k)\sqrt{n})$ rounds against any oblivious
adversary.  Since $k$-gossip can be completed in $nk$ rounds by
separately broadcasting each token over $n$ rounds, this yields a
bound of $\min\{nk, \tilde{O}((n+k)\sqrt{n})\}$ on centralized
$k$-gossip using token forwarding.

\onlyLong{
We begin by arguing that a $\tilde{O}(n^{3/2})$-round algorithm for
$n$-gossip implies a $\tilde{O}((n+k)\sqrt{n})$-round algorithm for
$k$-gossip.  We first make the assumption that $k$ is a multiple of
$n$.  When $k$ is not a multiple of $n$, then we add distinct dummy
tokens to make the total number of tokens a multiple of $n$.  Given
that the bound we seek is at least linear in $n + k$, this maintains
the asympototic complexity of the bound.  When $k$ is less than $n$,
we introduce $n-k$ dummy tokens.  When $k$ exceeds $n$, we group the
$k$ tokens into $\lfloor k/n \rfloor$ sets of $n$ tokens and one set
of less than $n$ tokens.  With this reduction, it is easy to see that
an $\tilde{O}(n^{3/2})$-round $n$-gossip algorithm implies an
$\tilde{O}((n+k)\sqrt{n})$-round $k$-gossip for arbitrary $k$.

}
\onlyShort{ In the full paper, we give a simple argument
that a $\tilde{O}(n^{3/2})$-round algorithm for $n$-gossip implies a
$\tilde{O}((n+k)\sqrt{n})$-round algorithm for $k$-gossip.}  We
  present our centralized algorithm for $n$-gossip in two parts.  We
  first solve a special case of $n$-gossip -- $n$-broadcast -- in
  which all the tokens are located in one node.  We then extend the
  claim to arbitrary initial distributions of the $n$ tokens.  We
  start by introducing two useful subroutines: \emph{random load
    balancing} and \emph{greedy token exchange}.

%We discuss the complexity of this problem in a distributed setting in
%Section~\ref{sec:load}.

\subsection{Random load balancing and greedy token exchange}
\label{sec:centralized.subroutines}
In the random load balancing subroutine, we have a set $F$ of nodes,
each of which contains the same set $T$ of at least $n$ \emph{items}
(each item is a copy of some token), and a set $R$ of nodes such that
$F \cup R$ is the set of all $n$ nodes.  The goal is to distribute the
items among nodes in $R$ such that the following properties hold at
the end of the subroutine: (B1) each item in $T$ is in exactly one
node in $R$; (B2) every node has either $\lfloor |T|/|R| \rfloor$ or
$\lceil |T|/|R|\rceil$ items; (B3) the set $X$ of items placed at any
subset $S \subseteq R$ of nodes is drawn uniformly at random from the
collection of all subsets of $T$ of size $\abs{X}$.

\noindent{{\bf LoadBalance}$(F,T,R)$:} Assign a rank to each item in
$T$ using a random permutation.  In round $i$, $i \in [|T|]$:
\begin{NoIndentEnumerate}
\item
Identify a node $v \in R$ that has been distributed fewer than
$\lfloor |T|/|R| \rfloor$ items yet, and is closest to a node in $F$,
say $v_0$, among all such nodes in $R$.
\item
Let $P$ denote a shortest path from $v_0$ to $v$.  Let $\ell$ be the
number of edges in $P$, and let $(v_{j-1}, v_j)$, $0 \le j < \ell$,
denote the $j$th edge in $P$; so $v_{\ell} = v$.  Then, $v_0$ sends
item of rank $i$ to $v_1$; in parallel, for every edge
$(v_{j-1},v_j)$, $1 \le j < \ell$, $v_{j-1}$ sends an arbitrary item
it received earlier in this subroutine to $v_j$.
\end{NoIndentEnumerate}

\begin{lemma}
\label{lem:centralized.balance}
The subroutine {\bf LoadBalance}$(F,T,R)$ completes in $|T|$ rounds
and satisfies the properties~(B1), (B2), and (B3).
\end{lemma}
\onlyLong{
\begin{proof}
The number of rounds taken by the subroutine is by construction.
Property~(B1) is satisfied since no copies of items are made.  In each
round, the number of items placed at exactly one node in $R$ increases
by 1, while the number of items at other nodes remains the same.  So
the total number of items placed at the nodes in $R$ at the end of $T$
rounds is exactly $|T|$.  Furthermore, no node in $R$ receives more
than $\lceil |T|/|R| \rceil$ items; this establishes property~(B2).
Finally, property~(B3) is satisfied since the items are placed in
order of a random permutation.
\end{proof}
} %onlyLong

The greedy token exchange is a one round subroutine in which the goal
is to maximize the number of new tokens received at each node in that
round.

\smallskip
\noindent{\bf GreedyExchange}: Fix a round.  For each node $v$, let
$S(v)$ be the set of tokens that node $v$ has at the start of the
round.  Let $N_v$ denote the set of neighbors of $v$.  Let $U_v$ be
the set $\cup_{u \in N_v} S(u) \setminus S(v)$.  For each node $v$, we
perform the following operations.  Construct a bipartite graph $H_v$,
in which one side is the set $N_v$, and the other side is the set
$U_v$. For each $u \in N_v$ and $\tau \in U_v$, there is a link
between $u$ and $\tau$ if token $\tau \in S(u)$.  Compute a maximum
bipartite matching $M_v$ in $H_v$.  If $(\tau, u)$ is in $M_v$, then
$u$ sends token $\tau$ to $v$.

\begin{lemma}
\label{lem:centralized.greedy}
In each round, the subroutine {\bf GreedyExchange} maximizes, for each
node $v$, number of new tokens that can be added to the node in that round.
\end{lemma}
\onlyLong{
\begin{proof}
Since each node can send a distinct token on each of its incident
edges, the problem of maximizing the number of distinct tokens
received by a node is independent of the same problem for a different
node.  By construction of the bipartite graph, for every possible set
$S$ of tokens arriving at a node $v$, we have a bipartite matching
$M_v$ in $H_v$ such that for every $\tau \in S$, there exists an edge
$(\tau, u)$ for some $u$ in $N_v$.  Similarly, every bipartite
matching $M_v$ corresponds to a valid set of token transfers to $v$ in
the network at that round.  Thus, {\bf GreedyExchange} maximizes, for
each node $v$, number of new tokens that can be added to the node in
that round.
\end{proof}
} 

\subsection{$n$-broadcast}
\label{sec:centralized.broadcast}
We now present a $\widetilde{\Theta}(n^{3/2})$-round algorithm for
$n$-broadcast, where all tokens are located initially in a single
node. The algorithm consists of $O(\log n)$ {\em stages}.  Let $U$
denote the set of all $n$ tokens.  We now describe each stage.  Call a
node {\em full}\/ if it has all of the $n$ tokens at the start of the
stage, and {\em non-full}\/ otherwise.  Let $R$ denote the set of
non-full nodes at the start of the stage, and let $r = |R|$.  The
stage consists of $\Theta(\sqrt{n}\log n)$ identical phases.  Each
phase consists of a sequence of steps divided into two segments:
distribution and exchange.

\onlyLong{
\begin{enumerate}
\item
}
\onlyShort{
\noindent
}
{\bf Distribution segment}: Distribute the $n$ tokens among the
non-full nodes $R$ in the network, as evenly as possible, in $n$
rounds by running {\bf LoadBalance}$(F, R, U)$.

\onlyLong{
\item
}
\onlyShort{
\noindent
}
{\bf Exchange segment}: Starting with the distribution of tokens as
specified in the preceding distribution segment; i.e., each full node
has all tokens, and each non-full node has exactly the tokens
distributed in the above segment, run $n$ rounds of {\bf
  GreedyExchange} maximizing the total number of new
tokens received by the nodes in each round.  
\junk{
Let
$S_i(u)$ denote the set of tokens that node $u$ has at the start of
round $i$.  In the $i$th round:
\begin{itemize}
\item
For each node $v$, we construct a flow network $F_v$.  Let $N_v$
denote the set of neighbors of $u$ in the dynamic network in round
$i$.  The vertex set of $F_v$ is the set $\{f_v\} \cup \{f_u: u \in
N_v\} \cup \{f_\tau: \mbox{ token } \tau\} \cup \{f_s\}$, where $f_s$
is a special source vertex.  There is an edge from $f_u$ to $f_v$ for
every edge $(u,v)$ that exists in the network in round $i$.  For every
token $\tau$ and every $u$ in $N_v$ such that $u$ has $\tau$ at the
start of round $i$, there is an edge from $f_\tau$ to $f_u$.  Finally,
for every token $\tau$, there is an edge from $f_s$ to $f_\tau$.
Every edge has unit capacity.  We compute a maximum flow from $f_s$ to
$f_v$.  In round $i$, we send a token $\tau$ from $u$ to $v$ if and
only if the maximum flow has a unit flow traversing the path $f_s
\rightarrow f_\tau \rightarrow f_u \rightarrow f)v$.
\end{itemize}
}
\onlyLong{
\end{enumerate}
}
%\end{NoIndentEnumerate}

\onlyLong{
\begin{lemma}
\label{lem:broadcast.phase}
If $R$ is the set of non-full nodes at the start of a stage, then
during any phase of the stage, the sum, over all nodes in $R$, of the
number of tokens received by the nodes is $\Omega(|R| \sqrt{n})$.
\end{lemma}
\begin{proof}
Fix a stage and a phase of the stage.  Consider the following initial
distribution of tokens at the start of the phase: each full node at
the start of the stage has all of the tokens, while each non-full node
has no token.  At the start of any round, we use {\em configuration}\/
to refer to the set of tokens that a node has at the start of the
round, starting from the preceding initial token distribution. 

Fix a round of the exchange segment.  We first consider the case in
which there exists a round in which the number of different
configurations at the start of the round, $m$, is less than
$\sqrt{n}$.  We number the $m$ configurations arbitrarily from $1$ to
$m$, and let $n_i$ denote the number of nodes in the $i$th
configuration.  By properties~(B1) and~(B2) of
Lemma~\ref{lem:centralized.balance}, after the distribution segment,
each non-full node started with $\lfloor n/r \rfloor$ or $\lceil n/r
\rceil$ distinct tokens that are unique among all non-full nodes.
Therefore, in any set of $\ell$ nodes that have the same
configuration, every node has at least $\ell \lfloor n /r \rfloor$
tokens.  Thus, the sum, over each node, of the number of tokens in the
node is at least $\lfloor n/r\rfloor \sum_{i = 1}^m n_i^2$, where
$\sum_i n_i = r$ and $m \le \sqrt{n}$.  Under these conditions,
$\sum_{i = 1}^m n_i^2$ is minimized when each $n_i = r/m$, yielding
the number of tokens received by the non-full nodes to be at least
$\lfloor n/r \rfloor m\cdot r^2/m^2 = \lfloor n/r \rfloor r^2/m$.  If
$r > n/2$, then we have $\lfloor n/r \rfloor r^2/m = r^2/m \ge r
\sqrt{n}/2$ since $m < \sqrt{n}$.  If $r \le n/2$, then we have
$\lfloor n/r \rfloor r^2/m \ge (n/r - 1)r^2/m \ge r\sqrt{n} -
r\sqrt{n}/2 = r\sqrt{n}/2$ since $m < \sqrt{n}$.  We thus have the
desired claim that the sum, over all nodes in $R$, of the number of
tokens received by the node is $\Omega(|R| \sqrt{n})$.

We next argue that if $m$ is at least $\sqrt{n}$, then either the
number of new token arrivals in the round, over all non-full nodes, is
$\Omega(\sqrt{n})$, or the total number of tokens already received by
the non-full nodes in this phase is $\Omega(r \sqrt{n})$.  Let $C_1$
through $C_m$ denote the $m$ configurations in this round.  We
construct an auxiliary graph in which each vertex is a configuration,
and we have an edge between vertices $C_i$ and $C_j$ if there is an
edge between a node having configuration $C_i$ and a node having
configuration $C_j$ in this round.  Note that since the network in
each round is connected, so is the auxiliary graph.  

Let ${\cal T}$ denote an arbitrary spanning tree in this auxiliary
graph.  Consider the set ${\cal S}$ of stars ${\cal T}$ formed by the
edges in either the odd levels of ${\cal T}$ or the even levels of
${\cal T}$, whichever is greater.  The number of edges in ${\cal S}$
is at least $\sqrt{n}/2$.  Let $S$ be a star in ${\cal S}$ with
configuration $C_0$ as the root and $C_i$ as the $i$th leaf.  Each
edge $(C_0, C_i)$ in ${\cal S}$ corresponds to an edge, say $(u_i,
v_i)$, where $u_i$ and $v_i$ hold configurations $C_0$ and $C_i$,
respectively, at the start of this round.  Furthermore, all the $v_i$s
are distinct nodes, while the $u_i$'s may not be distinct.  

We call an edge $(u_i, v_i)$ {\em bad} if $C_0 \subset C_i$ and $C_0$
contains all of the tokens that were at any node in configuration
$C_i$ {\em in the initial distribution of the Exchange segment};
otherwise, we call the edge good.  Note that if an edge $(u_i, v_i)$
is good then we have two cases: if $C_0$ is not a subset of $C_i$,
then we can transfer a token in $C_0 \setminus C_i$ from $u_i$ to
$v_i$; if $C_0$ does not contain all of the tokens that were at any
node in configuration $C_i$ in the initial distribution of the
Exchange segment, then we can transfer a token from $v_i$ to $u_i$
that is distinct from any other token that can be sent from $v_j \neq
v_i$ to $u_i$.  It thus follows that if the number of distinct
$u_i$'s, over all the stars in ${\cal S}$, is at least $\sqrt{n}/4$ or
the number of good edges is at least $\sqrt{n}/4$, then we can
identify a token transfer along the edges $(u_i, v_i)$ such that the
total number, over all the nodes in the stars, of the distinct tokens
received by the node in the round is at least $\sqrt{n}/4$.  By
Lemma~\ref{lem:centralized.greedy}, the exchange segment guarantees
that the number of token transfers in this round is at least
$\sqrt{n}/4$.

It remains to consider the case where the number of distinct $u_i$'s
is less than $\sqrt{n}/4$ and the number of good edges is less than
$\sqrt{n}/4$.  In this case, let $u$ be a node in a root configuration
$C_0$ of star $S$ that has bad edges to $\ell$ nodes, say $v_1$
through $v_\ell$, with configurations $C_1$ through $C_\ell$,
respectively.  By the definition of bad edges, we obtain that $u$ has
all the tokens that all the nodes with configuration $C_1$ through
$C_\ell$ had in their initial distributions.  Since $C_1$ through
$C_\ell$ is a superset of $C_0$, it follows that if $x$ is the number
of tokens in $\cup_{0 \le i \le \ell} C_i$, then each of the nodes
having any of configurations $C_0$ through $C_\ell$ has at least $x$
tokens.

We consider the following partition of all the nodes of the graph.
Let $S$ be a star in ${\cal S}$ with root $C_0$, the $i$th leaf being
given by $C_i$, and network edge $(u_i, v_i)$ corresponding to the
auxiliary graph edge $(C_0, C_i)$.  We define a group of the partition
to be the union of set of the nodes with configuration $C_0$ and the
union, over all $i$ such that $(u_i, v_i)$ is a bad edge, of the set
of nodes with configuration $C_i$.  If $(u_i, v_i)$ is a good edge, we
have the set of nodes with configuration $C_i$ form their own group.
Since the total number of distinct $u_i$'s, over all the stars in
${\cal S}$, is at most $\sqrt{n}/4$ and the number of good edges is at
most $\sqrt{n}/4$, the above procedure partitions all the nodes into
at most $\sqrt{n}$ groups such that in each group, every node has all
of the tokens that every node in its group had in the initial token
distribution of the Exchange step.  Using the same calculation as the
first case above, we obtain that in this case the total useful token
exchange already achieved is $\Omega(r \sqrt{n})$.
\end{proof}
} %onlyLong

\onlyLong{
\begin{lemma}
\label{lem:broadcast.progress}
After $\Omega(\sqrt{n} \log n)$ phases starting from a set $R$ of $r$
non-full nodes, there exist at least $r/3$ nodes in $R$ that receive
all tokens whp.
\end{lemma}
\begin{proof}
Suppose the number of non-full nodes remains at least $2r/3$ after $c
\sqrt{n} \log n$ phases, where $c$ is an arbitrary constant whose
value will be set later in the proof.  Then, by
Lemma~\ref{lem:broadcast.phase}, the sum, over all nodes, of the total
number of tokens received at the node during these phases is at least
$2c r n(\log n)/3$.  Note that since each phase is implemented
starting from an initial distribution in which every node not in $R$
has all tokens, while every node in $R$ has no tokens, the set of
tokens received by a node in a phase may intersect the set of tokens
received by the same node in another phase.

Since a full node does not receive any tokens during the distribution
and exchange segments, and any node receives at most $n$ tokens in a
phase, we obtain from an averaging argument that at least $r/3$
non-full nodes each receives at least $c n (\log n)/2$ tokens in this
stage; otherwise, the total number of token exchanges in $c \sqrt{n}
\log n$ phases is less than $(2r/3) c n (\log n)/2 + (r/3) c n (\log
n) = 2crn(\log n)/3$, a contradiction.

Consider any node $v$ that receives at least $2c n (\log n)/3$ tokens,
taken over all the $c\sqrt{n} \log n$ phases in this stage; note that
while the tokens received in a phase are distinct, these tokens are
not necessarily distinct across phases.  Suppose $v$ receives $p_i$
tokens in phase $i$.  By property~(B3) of
Lemma~\ref{lem:centralized.balance}, the $p_i$ tokens distributed to
$v$ in phase $i$ are drawn uniformly at random from the set of all
tokens.  Therefore, by a standard coupon collector argument, we obtain
that if $c$ is sufficiently large, $v$ has all of the $n$ tokens with
high probability, and thus becomes full after $c\sqrt{n} \log n$
phases.  This completes the proof that at least $r/3$ nodes that were
non-full at the start of the stage become full after $c \sqrt{n} \log
n$ phases.
\end{proof}
} 

\begin{theorem}
\label{thm:broadcast}
The $n$-broadcast problem completes in ${O}(n^{3/2}{\log^2{n}})$
rounds whp.
\end{theorem}
\onlyLong{
\begin{proof}
By Lemma~\ref{lem:broadcast.progress}, we obtain that there exist at
least $n/3$ nodes that have received all tokens after $O(n^{3/2} \log
n)$ rounds whp.  The remaining problem is that of
disseminating the $n$ tokens among $2n/3$ non-full nodes.  Applying
Lemma~\ref{lem:broadcast.progress} repeatedly $O(\log n)$ times
completes the proof of the theorem.  \junk{ The total number of rounds
  equals
\[
\sqrt{n} (n + n/2 + n/4 + \cdot + 1 + k) \log n.
\]
}
\end{proof}
}

\subsection{$n$-gossip}
\label{sec:centralized.gossip}
Our centralized algorithm for arbitrary $n$-gossip instances is as follows.

\noindent
{\bf Consolidation stage}: (a) For each token $i$, in sequence: for
$\sqrt{n}$ rounds, every node holding token $i$ broadcasts token $i$
(i.e., flooding of token $i$); (b) Identify a set $S$ of
$\tilde{O}(\sqrt{n})$ nodes such that every token is in some node in
$S$; arbitrarily assign each token to a node in $S$ that has the
token.

\noindent
{\bf Distribution stage:} Each node in $S$ makes $\sqrt{n}$ copies of
each of its allocated tokens, for a total of $n^{3/2}$ tokens in all,
including copies.  If any node in $S$ has a token multiset of fewer
than $n$ tokens, then it adds dummy tokens to the multiset to make it
of size $n$.  Let $T_u$ denote the multiset of tokens at $u$.  For
each node $u$ in $S$, we ensure that each node receives a distinct
random token from the multiset of $u$: {\bf LoadBalance}$(\{u\}, V,
T_u)$.

\junk{
The procedure is similar to the distribution segment
in Section~\ref{sec:centralized.broadcast}, and runs as follows.  Let
$X$ denote the multiset of tokens at $u$.
\begin{enumerate}
\item Identify a node $v$ that has been distributed fewer than
  $\lfloor |X|/n \rfloor$ tokens, and is closest to $u$, among all
  such nodes.
\item
Let $P$ denote a shortest path from $u$ to $v$.  Let $\ell$ be the
number of edges in $P$, and let $(v_{j-1}, v_j)$, $0 \le j < \ell$,
denote the $j$th edge in $P$; so $u = v_0$ and $v_{\ell} = v$.  Then,
$v_0$ sends token of rank $i$ to $v_1$; in parallel, for every other
edge $(v_{j-1},v_j)$, $1 \le j < \ell$, $v_{j-1}$ sends an arbitrary
token it received earlier in this distribution segment to $v_j$.
\end{enumerate}
}

\noindent
{\bf Exchange stage:} Maximize the number of token exchanges in each
round by repeatedly calling {\bf GreedyExchange}, until some node, say
$s$, has at least $n - c\sqrt{n}\log n$ tokens, for a constant $c$
that is chosen sufficiently large.  If $n$-gossip is not yet
completed, then: (a) Run $n$-broadcast with source $s$ to complete the
dissemination of the $n - c\sqrt{n} \log n$ tokens at $s$; (b) 
Run at most $c \sqrt{n} \log n$ separate broadcasts, spanning $n$
rounds, disseminating the remaining at most $c\sqrt{n} \log n$ tokens
to all nodes.
%\end{NoIndentEnumerate}

\onlyLong{
\begin{lemma}
\label{lem:gossip.consolidation}
The consolidation stage takes $n^{3/2}$ rounds, at the end of which we
can find a set $S$ of at most $O(\sqrt{n} \log n)$ nodes that together
contain all of the tokens whp.
\end{lemma}
\begin{proof}
The running time of the consolidation stage is immediate, since each
token broadcast period consists of $\sqrt{n}$ rounds.

Next, consider a set $S$ of $c\sqrt{n} \log n$ nodes selected uniformly
at random from the set of all nodes.  The probability that for a given
token $\tau$, $S$ does not include any of the at least $\sqrt{n}$
nodes that have $\tau$ after the consolidation phase is at most
\begin{eqnarray*}
\binom{n - \sqrt{n}}{|S|}/\binom{n}{|S|}  & = & \frac{(n - \sqrt{n})\cdots(n-\sqrt{n}-|S| + 1)}{n \cdots n-|S|+1}\\
& \le & \left(1 - \frac{\sqrt{n}}{n-|S|+1}\right)^{|S|}\\
& \le & \left(1 - \frac{2}{\sqrt{n}}\right)^{c\sqrt{n} \log n}\\
& \le & 1/\poly(n),
\end{eqnarray*}
for $n$ sufficiently large, and $c$ a sufficiently large constant.  By
applying a union bound over the $n$ tokens, we get the desired claim.
\end{proof}
}

\onlyLong{
\begin{lemma}
\label{lem:gossip.distribution}
Consider any phase of the distribution stage, in which node $v$
distributes its multiset $T_v$ of tokens among the $n$ nodes.  Each
node receives at least $\lfloor |T_v|/n\rfloor$ tokens from $T_v$, and
the tokens arriving at a node are a subset of $T_v$ drawn uniformly at
random from $T_v$.
\end{lemma}
\begin{proof}
Since each leader node assigns a random rank to each token copy it
has, by property~(B3) of Lemma~\ref{lem:centralized.balance}, the
token distribution process ensures that each node receives tokens of
random ranks.  Furthermore, by properties~(B1) and~(B2) of
Lemma~\ref{lem:centralized.balance}, each node receives $\lfloor
|T_v|/n \rfloor$ or $\lceil |T_v|/n \rceil$ tokens.
\end{proof}
\begin{lemma}
\label{lem:gossip.exchange}
After $O(n^{3/2} \log n)$ rounds of the exchange stage, we have a node
that has at least $n - \sqrt{n}\log n$ tokens whp.
\end{lemma}
\begin{proof}
Define a {\em configuration}\/ to be the set of tokens that a node has
at any time.  Let $m$ denote the number of distinct configurations
that are present at the start of any round.  If $m$ is at least
$\sqrt{n}/(c\log n)$, for a constant $c > 0$ chosen suitably large
later, then as in Lemma~\ref{lem:broadcast.phase}, we argue that the
sum, over all nodes, of the number of distinct tokens received by the
node in the round is at least $\sqrt{n}/(c\log n)$ (this follows from
Lemma~\ref{lem:centralized.greedy}, which establishes that {\bf
  GreedyExchange} maximizes the number of tokens exchanged in any
given round).  We can be in this case for at most $cn^{3/2} \log n$
rounds since each node receives at most $n$ distinct tokens.

In the second case, that is, where $m$ is at most $\sqrt{n}/(c \log
n)$, there exist at least $\sqrt{n} \log n$ nodes that have the same
set of tokens.  We now argue that any set of $c\sqrt{n} \log n$ nodes
together have at least $n - c\sqrt{n} \log n$ tokens, for $c$ chosen
suitably large.

Fix a set $X$ of $c \sqrt{n} \log n$ nodes.  Consider the $i$th phase
of the distribution stage in which tokens from a multiset set $T_v$
are distributed from a node $v$.  Let $\alpha_i$ be the number of
tokens in $\Gamma \cap T_v$.  By Lemma~\ref{lem:gossip.distribution},
the nodes in $X$ together receive a subset of $\lfloor |T_v|/n \rfloor
|X|$ tokens, chosen uniformly at random from the multiset $T_v$.  Note
that every token in $T_v$ has $\sqrt{n}$ copies in $T_v$ and $|T_v|$
is at least $n$; therefore, the probability that none of these
$\alpha_i$ distinct tokens in $\Gamma \cap T_v$ are in $X$ is at most
\begin{eqnarray*}
\frac{\binom{|T_v| - \sqrt{n}\alpha_i}{\lfloor |T_v|/n \rfloor |X|}}{\binom{|T_v|}{\lfloor |T_v|/n \rfloor |X|}}
& \le & \left(1 - \frac{\lfloor |T_v|/n \rfloor |X|}{|T_v| - \lfloor |T_v|/n \rfloor |X|}\right)^{\lfloor |T_v|/n\rfloor c \sqrt{n} \log n}\\
& \le & \left(1 - \frac{\sqrt{n} \alpha_i}{|T_v| - \lfloor|T_v|/n\rfloor c \sqrt{n} \log n}\right)^{\lfloor |T_v|/n\rfloor c \sqrt{n} \log n}\\
& \le & \left(1 - \frac{\sqrt{n}\alpha_i}{|T_v|}\right)^{c|T_v|\log (n)/\sqrt{n}}\\
& \le & e^{-c \alpha_i \log(n)/2}.
\end{eqnarray*}
(In the second last inequality, we use the fact that $|T_v|$ is at
least $n$, which implies $\lfloor |T_v|/n \rfloor \le |T_v|/(2n)$.  In
the last inequality, we use the fact that $1 - x \le e^{-x}$ for $0
\le x < 1$.)

Since the choice of the random permutation in each phase of the
distribution stage is independent of the choice in any other phase, we
obtain that the probability that none of the tokens in $\Gamma$ are in
$X$ is at most $\prod_i e^{-c \alpha_i \log(n)/2} \le e^{-c \sum_i
  \alpha_i \log(n)/2} = e^{-c^2 \sqrt{n} \log^2 (n)/2}$.  The number of
  different choices for $X$ and $\Gamma$ is $\binom{n}{c \sqrt{n} \log
    n}^2 \le (e\sqrt{n}/(c \log n))^{c\sqrt{n} \log n} \le e^{c\ln 2
    \sqrt{n} \log^2 n }$.  Applying a union bound, we achieve a high
  probability bound on the event that there exists a node with at
  least $n - c \sqrt{n} \log n$ tokens.

\junk{the probability that a particular
set $\Gamma$ of $c \sqrt{n} \log n$ tokens is not in $X$ at the end of
the distribution stage is at most $\prod_i (1 - \alpha_i/\sqrt{n})^{c
  \sqrt{n} \log n} \le e^{-c \sum_i {\alpha_i} \log n}$, where
$\alpha_i$ is the number of tokens in $\Gamma$ that participate in the
$i$th distribution step of the distribution phase.  Since $\sum_i
\alpha_i = c \sqrt{n} \log n$, we obtain that the preceding
probability is at most $e^{-c^2 \sqrt{n} \log^2 n}$.  
}
\end{proof}
} 

\begin{theorem}
Our centralized algorithm completes in $O(n^{3/2} \log^2 n)$ rounds, whp.
\end{theorem}
\onlyLong{
\begin{proof}
The consolidation stage takes $O(n^{3/2} \log n)$ rounds.  The
distribution stage takes $O(n^{3/2})$ rounds.  The exchange stage
takes $O(n^{3/2} \log^2 n)$ rounds for the $n$-broadcast (by
Theorem~\ref{thm:broadcast}) and $O(n^{3/2}\log n)$ rounds for
broadcasting the last $O(\sqrt{n} \log n)$ tokens.  The high
probability successful completion follows from
Lemma~\ref{lem:gossip.exchange}, the correctness of $n$-broadcast
(Theorem~\ref{thm:broadcast}), and part~(b) of the exchange stage.
\end{proof}
} 

\section{Concluding remarks}
\label{sec:concs}
Our work has focused on the basic question of whether there exists a
fully distributed $n$-gossip protocol that runs in sub-quadratic time,
i.e., $O(n^{2-\epsilon})$ rounds (for some positive constant
$\epsilon$), or even faster $O(n \ \polylog n)$ rounds, under an
oblivious adversary.
%The same question is reasonably well understood under a strongly adaptive adversary: we know it is not
%possible to achieve such a sub-quadratic bound.
% (ignoring poly-logarithmic factors).  
We showed that somewhat surprisingly, \randdiff, a potentially strong
candidate for a fast distributed algorithm, has a
$\tilde{\Omega}(n^{3/2})$ lower bound.  Morever, for symmetric
knowledge-based algorithms (SKB), we showed a lower bound of
$\tilde{\Omega}(n^{4/3})$ under invasive adversaries, a stronger
version of oblivious adversaries.  We complemented these results with
two upper bounds.  First, we showed that \randdiff\ can complete
$n$-gossip in subquadratic time --- $\tilde O(n^{5/3})$ --- under a
restricted oblivious adversary which has to respect some
infrastructure-based path constraints.  We believe the analysis of
\randdiff, in fact, extends to the more efficient \symdiff\ protocol
as well, and also to more general path-respecting adversaries.
Second, we presented a centralized algorithm that achieves a $\tilde
O(n^{3/2})$ bound using under any oblivious adversary.

Our work leaves several intriguing open problems and directions for
future research: Is there a hybrid of \randdiff\ and a knowledge-based
algorithm that can achieve sub-quadratic complexity?  What is the best
bound for $n$-gossip achieved by centralized token-forwarding?
Explore paths-respecting and related models further to gain a better
understanding of network dynamics from a practical standpoint.

\junk{
These, to our best
of the knowledge, are the first sub-quadratic bounds for
token-forwarding algorithms under an oblivious adversary.}
%It is tempting to conjecture that a
%$\tilde{\Theta}(n^{3/2})$ fully distributed algorithm exists and may
%even be optimal.  Perhaps, an algorithm that uses features of both
%\randdiff and SKB might prove helpful.  Showing tight bounds on the
%paths-respecting model is also an interesting open problem.

\bibliographystyle{acm}
\bibliography{tokenforwarding,dynamic,RR}

\end{document}